\documentclass[journal,twoside]{IEEEtran}
\usepackage{amsmath,amsfonts}
\usepackage{algorithmic}
\usepackage{algorithm}
\usepackage{array}
\usepackage[caption=false,font=normalsize,labelfont=sf,textfont=sf]{subfig}
\usepackage{textcomp}
\usepackage{stfloats}
\usepackage{array}
\usepackage{makecell}
\usepackage{url}
\usepackage{verbatim}
\usepackage{graphicx}
\usepackage{cite}
\usepackage{orcidlink}
\usepackage{amsthm}
\usepackage{amssymb}
\usepackage{xcolor}
\usepackage{soul}
\theoremstyle{plain}
\newtheorem{theorem}{Theorem}
\newtheorem{lemma}{Lemma}
\newtheorem{corollary}{Corollary}
\newtheorem{proposition}{Proposition}

\theoremstyle{definition}
\newtheorem{assumption}{Assumption} 
\newtheorem{definition}{Definition}

\newtheorem{problem}{Problem}

\theoremstyle{remark}
\newtheorem{remark}{Remark}

\begin{document}

\title{Optimal Calibration of Quantum Network Links}

\author{Vinay Kumar~\orcidlink{0000-0002-4635-3237}, 
Claudio Cicconetti~\orcidlink{0000-0003-4503-4223}, 
Marco Conti~\orcidlink{0000-0003-4097-4064}, 
and Andrea Passarella~\orcidlink{0000-0002-1694-612X}%
\thanks{Vinay Kumar is with the Department of Information Engineering, University of Pisa, and also with the Institute for Informatics and Telematics (IIT), National Research Council (CNR), Pisa, Italy (e-mail: vinay.kumar@phd.unipi.it).}%
\thanks{Claudio Cicconetti, Marco Conti, and Andrea Passarella are with the Institute for Informatics and Telematics (IIT), National Research Council (CNR), Pisa, Italy (e-mails: \{c.cicconetti, marco.conti, a.passarella\}@iit.cnr.it).}%
}

\maketitle

\begin{abstract}
The reliable distribution of entanglement is essential for the effective operation of quantum networks. Due to fundamental differences between quantum and classical communication systems, it is necessary to develop specialised algorithms and protocols that also account for quantum-specific constraints.
In this work, we focus on the issue of recalibration.
As suggested by recent experimental studies, the process of local entanglement generation in a quantum link degrades over time due to environmental changes that have to be estimated and compensated via a fixed-duration calibration operation, during which the link is not available.
Therefore, in such a quantum network, every link alternates between an activation period, during which it operates normally, and a calibration period, during which it cannot participate in the end-to-end entanglement distribution, thereby creating a trade-off between link quality (the initial fidelity of generated pairs, which decays during activation) and availability (the fraction of time the link is usable, which calibration reduces). We develop analytically a protocol for optimally assigning activation periods to each link in linear quantum repeater chains, subject to any general end-to-end fidelity requirements and local initial fidelity thresholds. Building on this foundation, we
extend to general quantum networks, where multiple paths may cross at common links. To this, we propose a heuristic approach, evaluate it through simulations, and compare its performance with a benchmark, a numerical approach and theoretical bounds, confirming the analytical ordering of their performance.
 
\end{abstract}


\section{Introduction}\label{sec:intro}
\IEEEPARstart{Q}{uantum} networks, envisioned to operate alongside the classical internet, promise transformative capabilities in enhancing security, computation, and enabling specialised applications like quantum sensing \cite{vkthesis, vk2025quantuminternet, rohde2025quantuminternet, meddeb2025quantuminternet}. Realising this vision requires protocols and systems that explicitly also account for quantum-specific constraints, most notably no-cloning, photon indistinguishability (including between independently generated photons, as recently demonstrated \cite{lalindis25}), which enables Bell-state measurement (BSM), and the fragility of entanglement under noise and drift.

At the core of wide-area quantum networking lies \emph{entanglement distribution}: entanglement is first generated between neighbouring nodes and then extended across the network via \emph{entanglement swapping} at intermediate nodes \cite{zukowski1993swapping}. Recent experimental milestones demonstrate the feasibility of multi-hop and multinode operation across different platforms, including remote solid-state qubits linked into a multinode network \cite{pompili2021multinode}, trapped-ion fiber links over hundreds of meters \cite{viktor2023trappedionentanglement}, telecom-compatible nanophotonic quantum memories \cite{knaut2024nanophotonicentanglement}, and city-scale automated entanglement distribution over deployed fiber \cite{craddock2024newyork}. Entanglement-assisted functionality at the network level has also been demonstrated, e.g., teleportation between non-neighbouring nodes \cite{hermans2022teleportation}.
However, these advances overlook practical aspects that could later reveal system bottlenecks in large-scale quantum network deployments, possibly limiting their impact and commercial exploitation. 
In this paper, we focus on one of these potential issues: \emph{calibration}.

In fiber-connected links, slow polarisation and spectral drift degrade the initial fidelity of the generated entanglement along the link over time, requiring periodic recalibration during which links are temporarily unavailable. A representative dataset is shown in Fig.~\ref{fig:fid_vs_time}, derived from the trapped-ion fiber experiment of \cite{viktor2023trappedionentanglement}. The average transmitted fidelity falls from $\approx 1$ to $\approx 0.5$ over $\sim 10{,}000$~s ($\sim$166~min), after which a $\sim$2~min calibration restores the link. Such duty-cycling directly reduces end-to-end throughput and crucially interacts with multi-path orchestration when links are shared among paths.
While experimental platforms can handle this via manual calibration routines, a principled network-layer optimisation of activation and calibration scheduling under fidelity constraints remains underexplored. Existing work as discussed in Sec.~\ref{sec:relatedwork} either treats calibration as an operational routine without optimising it, or optimises entanglement distribution over links that are assumed to be always available. This is the problem we aim to address in this work.

\begin{figure}[!t]
\centering
\includegraphics[width=\columnwidth]{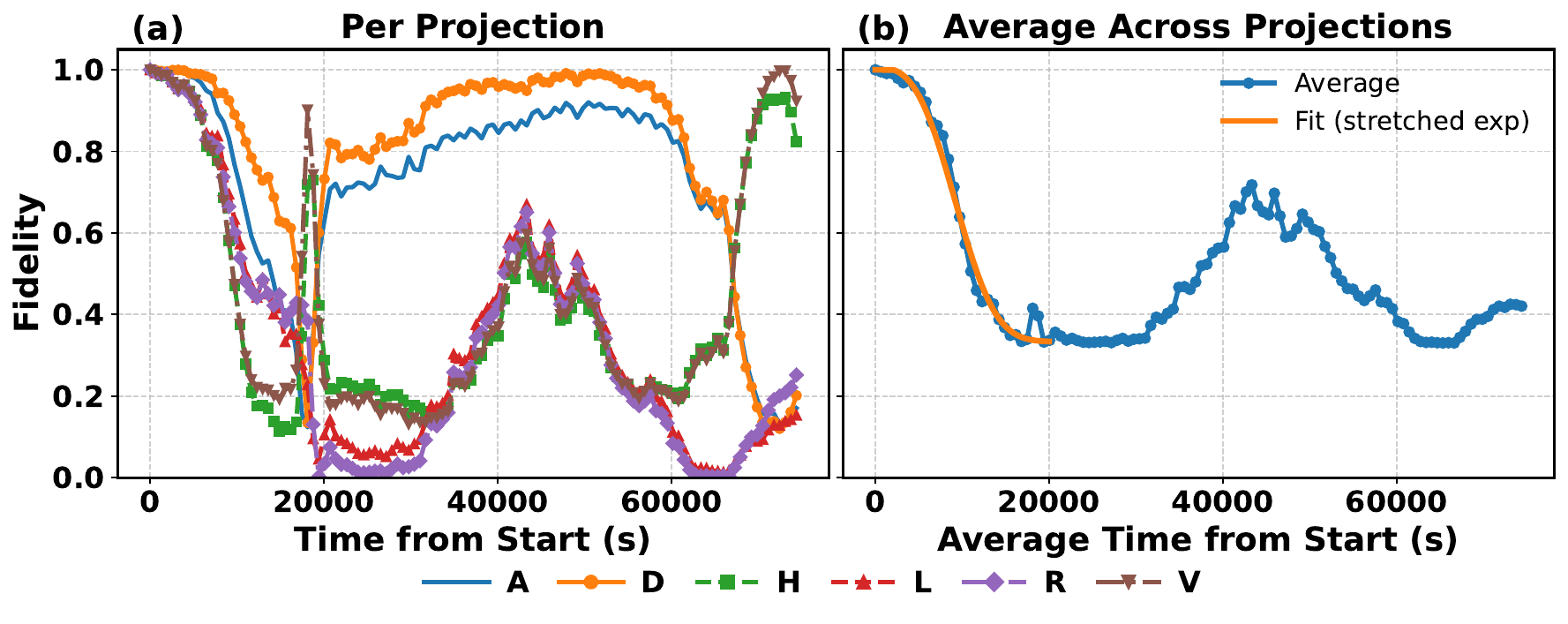}
\caption{Experimental polarisation/fidelity drift on the fiber link of \cite{viktor2023trappedionentanglement}. 
(a) Time evolution of the six polarization fidelity $\{A, D, H, L, R, V\}$. 
(b) The averaged fidelity decay over the activation period, fitted with an exponential-family model selected by Akaike Information Criterion (AIC). Refer to Appendix~\ref{appendix:experimental_data} for a detailed processing of experimental data and derivation of this figure.}
\label{fig:fid_vs_time}
\end{figure}

In this work, we formulate calibration-aware entanglement distribution as a network optimisation problem and make the following contributions:
\begin{enumerate}
    \item We analytically derive the optimal calibration schedule for linear quantum network chains, termed \emph{Quantum Link Orchestration (QLO)}, based on a fidelity decay model validated in recent experiments. QLO balances activation time against calibration downtime to maximise throughput under fidelity constraints.
    \item We then extend QLO to general topologies by first establishing a benchmark as the minimum selection rule (MIN), followed by \emph{Greedy quantum link Orchestration} (GRO) heuristic, which shows improvement over the benchmark.
    \item We simulate the problem and compare the performances of five approaches, which include relevant reference, upper bound, numerical approach, and ultimately showcase the results.
\end{enumerate}
Our results provide a calibration-aware foundation for scheduling and resource allocation in near-term quantum networks, aligned with current architectures and experimental constraints \cite{vkthesis, vk2025quantuminternet, rohde2025quantuminternet, meddeb2025quantuminternet}.

The organisation of the paper is as follows. In Section~\ref{sec:relatedwork}, we discuss some related works. In Section~\ref{sec:sys_model} we establish the system model. In Section~\ref{sec:calibration_prelim}, we discuss optimal calibration in a linear quantum chain. In Section~\ref{sec:calibration_general}, we discuss the quantum links orchestration in general network topologies. In Section~\ref{sec:simulation}, we show the simulation results and performance evaluation of the evaluated approaches and Section~\ref{sec:conclusion} concludes this work.

\section{Related Work}\label{sec:relatedwork}

This work sits at the intersection of (a) experimental and systems work on operating fiber-based and repeater-connected quantum links under drift and calibration requirements, and (b) algorithmic work on routing, scheduling, and resource allocation under fidelity constraints.

In field deployments, entanglement distribution over optical fiber must contend with time-varying polarisation and spectral drift, motivating active alignment and periodic calibration procedures. A line of work studies polarisation alignment and compensation mechanisms for quantum links and networks, aiming to reduce overheads and improve robustness in deployed settings \cite{peranic2023polcomp, dowling2023nonlocalpol, zhou2025polcomp}.
At the system level, recent quantum networking testbeds integrate calibration as a first-class operational routine within network control loops. The IEQNET architecture explicitly incorporates basis alignment and channel calibration steps under a central controller before and during entanglement distribution sessions, and re-initiates such procedures periodically as part of service execution \cite{chung2022ieqnet}.
Similarly, QUANT-NET develops a multi-node testbed and a control stack that separates node-level real-time control from network-level orchestration, including calibration and tuning tasks required for stable operation of trapped ion networking components \cite{schon2024quantnet}.
Campus-scale orchestrators such as ArQNet further highlight the need for automation of operational workflows, including routines that maintain stable entanglement distribution over extended runs \cite{islam2025arqnet}. \textit{These efforts motivate our central modelling choice, that is, calibration is not merely a background engineering detail but a performance-limiting duty cycle that interacts with multi-hop and multi-path operation.} Existing works largely treat calibration as an operational routine, whereas we model activation duration between calibrations as an explicit optimisation variable.

A substantial literature studies how to deliver end-to-end entanglement under fidelity constraints via routing and protocol selection, often incorporating purification and repeater operation models \cite{pant2019routing, vk2025endtoend, vk2024mixed, azuma2023rmp}.
Recent work has emphasised network-level scheduling of purification resources for concurrent demands, showing that fidelity constraints and shared resources can substantially reduce throughput if not coordinated \cite{xiao2024psc}.
Other approaches enforce fidelity constraints by selecting link configurations or operating modes as part of entanglement routing, thereby coupling per-link decisions with end-to-end requirements \cite{zhang2025linkconfig}.
This work provides a complementary study in two ways. First, \textit{we focus on a distinct but experimentally grounded control knob: the activation duration between calibrations}, which directly determines the initial fidelity of the distributed entangled pair along links over time. Second, we provide a structural optimisation view that yields closed-form solutions on linear chains and enables network-wide optimisation under shared link contention.

Several works adapt classical resource allocation ideas to quantum networks, including NUM-style frameworks that jointly value rate and quality. Quantum Network Utility Maximisation (QNUM) formulations propose concave utilities based on entanglement measures and enable distributed primal-dual style control \cite{vardoyan2022qnum}.
Follow-up work develops decentralised control architectures and stability analysis for distributed resource allocation in quantum networks \cite{panigrahy2025distributed}.
Related optimisation perspectives also connect routing to fairness and multi-commodity resource allocation objectives \cite{wang2025fair}. Calibration-aware resource allocation has also been considered in hub or switch-style settings. In particular, \cite{gauthier2024resourceallocation} models an on-demand resource allocation algorithm for a quantum network hub and explicitly incorporates calibration periods into the service dynamics, providing performance analysis complementary to our network-wide optimisation view.
These formulations are complementary to ours: \cite{vardoyan2022qnum, panigrahy2025distributed} establish how concave, separable utilities of rate and quality can be optimised through distributed control, and the QLO allocation we derive is exactly such a separable concave objective in log-fidelity space. We instead fix routing and optimise the temporal calibration knob, which sits orthogonal to the routing-and-fairness layer of \cite{wang2025fair}.

Finally, recent work has argued that temporal aspects such as entanglement freshness should be treated explicitly in scheduling and control, proposing metrics such as fidelity age and corresponding schedulers that reduce extreme staleness events while preserving throughput \cite{ercetin2026fidelityage}.
\textit{This reinforces the broader message that time-structured effects matter in quantum network operation.} Our contribution addresses a complementary temporal mechanism, namely drift within activation periods and its mitigation via optimal calibration scheduling.

In summary, the experimental
and systems literature establishes calibration as a recurring operational
requirement but never treats its period as a design variable. On the other hand, the
algorithmic literature optimises routing, purification and rate allocation
over links which are assumed to be permanently available. To the best of our
knowledge, no prior work treats the activation duration between
calibrations as the optimisation variable, subjected to local
and end-to-end fidelity constraints simultaneously and in turn contention on links shared by
several paths. This is the gap addressed here.

\begin{table*}[!t]
\centering
\caption{Summary of notation used in this paper.}
\label{tab:notation}
\renewcommand{\arraystretch}{1.15}
\setlength{\tabcolsep}{6pt}
\begin{tabular}{p{0.22\textwidth} p{0.73\textwidth}}
\hline
\textbf{Symbol} & \textbf{Meaning} \\
\hline
$e$ & A physical quantum link (edge) \\
$a_e$ & Activation period for link $e$ \\
$a_e^{th}$ & Threshold on activation period duration of link (due to $F_e^{th}$)\\
$c$ & Calibration period duration of link $e$ (same fixed for all links) \\
$T_a$ & Length of a complete cycle of the activation period: duration during which fidelity of generated entanglement is non-zero.\\
$p,q$ & Success probabilities for external links and internal operations, respectively \\
$\mathcal{T}$ & End-to-end throughput (entanglement delivery rate), $\mathcal{T}_{\pi}$ for a path $\pi$ \\
$N$ & Number of links (external) in a linear chain (path length in links), $N_{\pi}$ for a path denoted by $\pi$ \\
$\mathcal{A}$ & Number of attempts (model constant) \\
$x$ & Number of attempts in a complete activation period cycle $(T_a)$, typically $x = T_a/T_s$ \\
$T_s$ & Duration of one entanglement generation attempt \\
$F_{ete}$ & End-to-end fidelity after entanglement swapping along a path \\
$F_e$ & Initial fidelity of an entangled pair generated on link $e$ \\
$F_e^M$ & Best achievable initial fidelity on link $e$ (value as $a_e \to 0$), default numerical value $\approx 0.999$ \\
$\Gamma$ & Fidelity decay rate parameter for link $e$ (exponential drift model) \\
$p_1,p_2$ & Single qubit and two qubit operation efficiency parameters in swapping model \\
$\eta$ & Bell state measurement efficiency parameter \\
$F_e^{th}$ & Per-link threshold on initial fidelity (minimum acceptable initial fidelity) \\
$F_{ete}^{th}$ & End-to-end fidelity requirement (threshold) for a path \\
$\mathcal{U}$ & Noisy swapping operation factor, $\mathcal{U}=(p_1^2p_2)^{N-1}\left(\frac{4\eta^2-1}{3}\right)^{N-1}$ \\
$F_i$ & Initial fidelity function of link at $i^{th}$ position in a chain (generic notation) with corresponding threshold $F_i^{th}$\\
$L_e$ & Log fidelity variable, $L_e = \Gamma a_e = \ln\!\left(\frac{A}{F_e-1/4}\right)$ \\
$L_e^{th}$ & Per link log fidelity upper bound induced by $F_e^{th}$ \\
$A$ & Drift model constant, $A = F_e^M - \tfrac{1}{4}$ \\
$\Omega$ & Equal activation optimum for a path (from $F_{ete}^{th}$), $\Omega_F$ (F-space) $\&$ $\Omega_L$ (L-space) \\
$\pi$ & A generic source-to-destination path \\
$\Omega_\pi$ & Optimal point for a specific path $\pi$ \\
$K$ & Calibration term in $L$ space, $K=\Gamma c$ \\
$\mathcal{L}$ & End-to-end log fidelity budget for a chain, typically $\sum_{e\in\pi} L_e \le \mathcal{L}$ (emerges from $F_{ete}^{th}$) \\
$L_i$ & Initial fidelity function of link in L-space at $i^{th}$ position in a chain (generic notation) with corresponding threshold $L_i^{th}$\\

$i,j,k$ & Link indices; also used for partition sizes in QLO (see Theorem~\ref{theorem.chain}) \\
$\gamma \in \mathbb{N}$ & Optimization level (recursion depth) in QLO \\
$P$ & Set of paths among sources and destinations in a general quantum network graph \\
$\mathcal{L}_\pi$ & Path specific budget for path $\pi$ in the network wide formulation \\
$\Pi_e$ & Set of path transversing link $e$\\
$T$ & Sum of path throughputs over all paths between sources and destinations of a combination, typically $T = \sum_{\pi \in P} \mathcal{T}_{\pi}$\\
$L_{\{\pi, e\}}$ & Log-fidelity allocation of link $e$ of path $\pi$ on individual treatment of each path with QLO between sources and destinations\\
$M$ & Sum of links between sources and destinations counted path by path for a quantum network graph\\
$\Delta_\pi$ & The residual budget left for path $\pi$ after allocation using QLO\\
$r_e$ & The headroom left for a link $e$ after the allocation using QLO (refer Algorithm~\ref{alg:gro_hybrid})\\
$U_\pi$ & Set of unshared links in path $\pi$ (utilised in Algorithm~\ref{alg:gro_hybrid})\\
$\Delta\%$ & Percentage of an approach sum of path throughputs ($T$) relative to QLO (used in Section~\ref{ssec:relative_performance})\\
\hline
\end{tabular}
\end{table*}

\section{System model}\label{sec:sys_model}
Throughout this work, we consider a general quantum network structure consisting of quantum repeaters, quantum links, quantum devices, and a network controller \cite{vk2025endtoend, vk2024mixed}. The components of this quantum network consist of the necessary equipment, such as heralded sources of entanglement, detectors, and quantum memories, enabling execution of any network-level protocol. The quantum links are assumed to be optical fibers connecting neighbouring nodes, akin to the experimental setting in \cite{viktor2023trappedionentanglement}, and are continuously required to attempt the generation of an entangled pair, either for advanced or on-demand entanglement generation required for quantum routing. As discussed in Section~\ref{sec:intro}, we model the degradation in the initial fidelity of a successfully generated entanglement pair along with a calibration period, as shown in Fig.~\ref{fig:fidelity_decay}.
The relevant notation used in the paper is reported in Table~\ref{tab:notation}.

\definition{In a quantum network, for a quantum link $e$, the link activation period $(a_e)$ is the period during which it is actively generating entanglement.}
\definition{In a quantum network, for a quantum link $e$, the link calibration period $(c)$ is the period dedicated to calibrating the quantum link to rejuvenate against drift and enable optimal entanglement generation. The calibration period is a fixed constant, since operationally, it takes a fixed amount of time to calibrate any link irrespective of its prior condition.}

The generation of the entangled pair through each link is inherently probabilistic. For a successful delivery of throughput, each of the links in the path connecting end users has to be successful within the complete activation period cycle ($T_a$), which happens with probability $p$, as depicted in Fig.~\ref{fig:fidelity_decay}. Furthermore, all the entanglement swapping operations at intermediate quantum repeaters, each with probability $q$, have to be successful. Without loss of generality, for simplicity of notation, we consider that these probabilities $p,q$ are the same for all the links/nodes in the network.
Calibration restores only the \emph{initial fidelity} of the entangled pairs and does not change the underlying attempt success probabilities $p$ and $q$. Thus, drift affects the quality of successful pairs but not the raw attempt rate.

\subsection{Introduction to calibration and its modelling}\label{ssec.calibration_modeling}
In this section, we introduce the concept of calibration, its types, and throughput modelling. As illustrated in Fig.~\ref{fig:fidelity_decay}.b, the calibration of a single link affects the overall operation of the linear quantum network. In particular, even if one link is operating above the threshold while another link reaches the threshold, successful entanglement generation on the former cannot contribute to the throughput of the linear quantum network, since end-to-end operation requires successful entanglement generation across both links.
\definition{In a quantum network, the throughput $(\mathcal{T})$ is the rate at which quantum states can be successfully transmitted across the two endpoints.}

In a functioning quantum network, the calibration can be thought of as performed in two ways.
\begin{enumerate}
    \item \textbf{Synchronised periodic calibration (SPC):} In SPC, all the quantum links are in sync to have overlapping activation and calibration periods (visualise the scenario in Fig.~\ref{fig:fidelity_decay}.a with a fixed calibration period as soon as fidelity hits the threshold on each of the links). This approach is simple to implement as it only requires the nodes to adopt a periodic maintenance schedule, but the entire network becomes inactive for customers at regular intervals. Also, it may be inefficient since the system period must be chosen based on the link with the stringent calibration requirements.
    \item \textbf{Unsynchronised periodic calibration (USPC):} In USPC, the quantum links operate with non-overlapping activation and calibration periods (as the scenario in Fig.~\ref{fig:fidelity_decay}.b). By adapting the calibration period link by link, this approach is more flexible than SPC, and it allows the network to remain operational at all times, but requires a careful assignment of the activation periods ($a_e \leq T_a$) constraint by local fidelity threshold and overall end-to-end fidelity threshold.

\end{enumerate}

\begin{figure}[!t]
\centering
\includegraphics[width=\columnwidth]{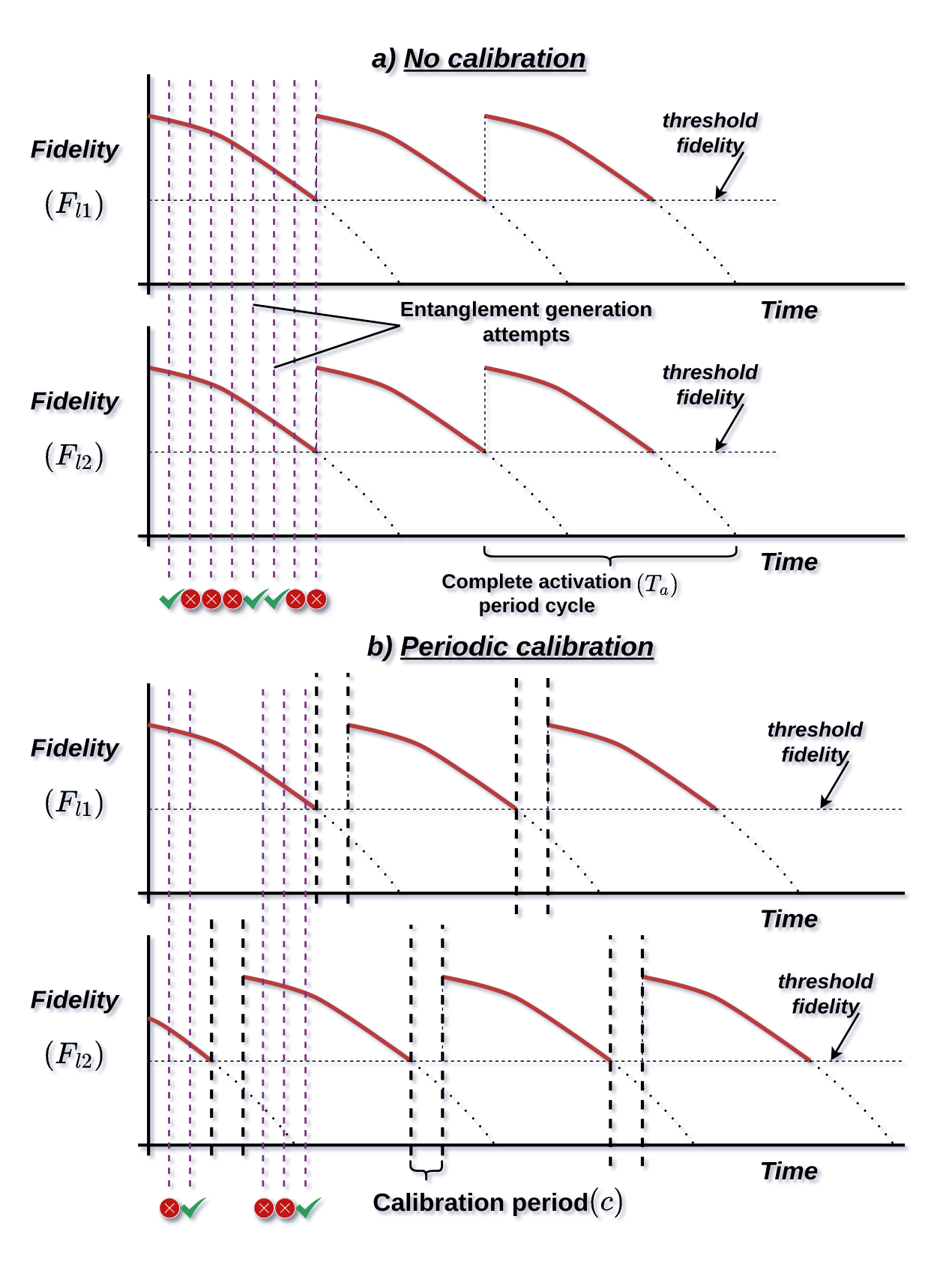}
\caption{An example of a two-link ($F_{l1}$ and $F_{l2}$) linear quantum network. \textbf{a) No calibration:} Assuming both links are synchronised, the initial fidelity of entanglement generation (provided it is successful) decays during the complete activation period cycle ($T_a$). Operationally infeasible, as it requires instantaneous rejuvenation of each link as soon as fidelity hits the threshold. \textbf{b) Periodic calibration:} Assuming independent link operation, as soon as fidelity hits the threshold, a fixed calibration period ($c$) is required for each link to rejuvenate. The $\checkmark$ in the figure means the successful delivery of an entangled pair across the network, while $\times$ denotes an unsuccessful attempt due to probabilistic failure of entangled pair generation on either or both of the links involved.}
    \label{fig:fidelity_decay}
\end{figure}

For a linear quantum network chain with $N$ links, the throughput with \textbf{no calibration} (scenario Fig.~\ref{fig:fidelity_decay}.a) is~\cite{pant2019routing}:
\begin{equation}\label{eq:nc_throughput}
     \mathcal{T}_{nc} = \mathcal{A} p^N q^{N-1},
\end{equation}
where $\mathcal{A}$ is an arbitrary number of total attempts. An attempt is only successful if local generation succeeds on each of the $N$ links and entanglement swapping succeeds on each of the $N-1$ intermediate nodes.

For SPC, the throughput decreases by a factor equal to the ratio of the activation period to the total cycle time, which accounts for the fact that the network is only operational during the activation period. Thus, the throughput of SPC is given by:
\begin{equation}\label{eq:spc_throughput}
    \mathcal{T}_{spc} = \mathcal{A} \ p^N q^{N-1} \left( \frac{x \ T_s}{x \ T_s + c} \right),
\end{equation}
where $x$ is the number of attempts possible in the complete activation period cycle ($T_a$), $T_s$ is the time taken by an attempt, and $c$ is the calibration period.

For USPC, considering independent calibration of links, the only time when the end-to-end path is operational is when all links are simultaneously active, which happens with probability:
\begin{equation}
    \mathcal{P}_{uspc} = \left( \frac{x}{ x + \frac{c}{T_s}} \right)^{N},
\end{equation}
with a corresponding throughput of 
\begin{equation}\label{eq:uspc_throughput}
    \mathcal{T}_{uspc} = \mathcal{A} \ p^N q^{N-1}  \left( \frac{ x \ T_s}{x \ T_s + c} \right)^{N}.
\end{equation}

From Eq.~\eqref{eq:nc_throughput}, \eqref{eq:spc_throughput} \& \eqref{eq:uspc_throughput} we have the ratio of throughput as:
\begin{equation}
\mathcal{T}_{nc} : \mathcal{T}_{spc} : \mathcal{T}_{uspc} =
\begin{aligned}[t]\label{eq:ratioeq}  1 :
\left( \frac{x\,T_s}{x\,T_s + c}
\right) : \left(
    \frac{x\,T_s}{x\,T_s + c}
  \right)^{N}.
\end{aligned}
\end{equation}
From the above equation, it is clear that approaches not accounting for
calibration overstate throughput as compared to SPC and exponentially
overestimate compared to USPC. The no-calibration case corresponds to the
limit $c \to 0$, in which both ratios reduce to unity: every link is
implicitly assumed to rejuvenate instantaneously, incurring zero calibration
downtime. This is precisely the assumption made by any protocol that ignores
calibration, and it is infeasible in practice, since recalibration requires a
finite, non-zero time (Fig.~\ref{fig:fidelity_decay}). Note from the ratio above that SPC attains a higher
nominal throughput than USPC for the same activation period. Since USPC pays a
coincidence penalty, that is, the path delivers only when all $N$ links are
simultaneously active. The motivation for studying USPC is therefore
operational rather than throughput-driven. SPC imposes a single network-wide
calibration cycle whose period must be set by the most stringent link,
$a=\min_e(a_e^{th})$, forcing links with looser fidelity thresholds to
recalibrate far more often than necessary and taking the entire network offline
during every calibration. USPC removes both limitations by letting each link
adapt its own activation period while keeping the network continuously
operational, at the cost of a harder scheduling problem, that is, the per-link
assignment of activation periods, which is the main focus of this work.

\subsection{Network throughput and fidelity modelling}
We now utilise the calibration modelling developed in Sec.~\ref{ssec.calibration_modeling} for linear network throughput modelling and additionally establish fidelity modelling.
\definition{In a quantum network, end-to-end fidelity $(F_{ete})$ is defined as the fidelity of the long-distance entangled state produced between two endpoints through a series of entanglement swapping operations along the network path.}

\begin{assumption}\label{assumption.activationcalibrationratio}
    For a quantum network with a single link $e$, the effective throughput $\mathcal{T}$ is proportional to the ratio of activation to the total cycle time as in Eq.~\eqref{eq:spc_throughput} $$\mathcal{T} \propto \left( \frac{a_e}{a_e + c} \right),$$
    where $a_e = x \ T_s$ is the activation period of the link.
\end{assumption}
In this work, we will formulate and discuss the general case of USPC with independent calibration. The simpler case of SPC follows by setting $N=1$ (see Eq.~\eqref{eq:ratioeq}).

\begin{lemma}\label{lemma.throughputN}
    In a linear quantum network with $N$ links, the effective throughput across the network where each link $e$ is operating independently is given by (Eq.~\eqref{eq:uspc_throughput}), i.e.: $$\mathcal{T} = \mathcal{T}_{nc} \left( \frac{a_e}{a_e + c} \right)^N,$$
    where $\mathcal{T}_{nc}$ (Eq.~\eqref{eq:nc_throughput}) encompassing the successful establishment of nodes (with probability $q$) and links (with probability $p$) of the path.
\end{lemma}
\begin{proof}
    From Assumption~\ref{assumption.activationcalibrationratio}, the throughput of a single link is proportional to the time it spends in the activation period relative to the total cycle time. Extending this to $N$ links, which are operating independently, gives the required relation.
\end{proof}
\begin{remark}\label{remark.c_e}
    The establishment of nodes and links in Lemma~\ref{lemma.throughputN} refers to the probabilistic success of entanglement swapping within nodes and the success of initial entanglement generation over links. While nodes and links are always physically present, it is this probabilistic establishment that makes them usable.
\end{remark}

\begin{assumption}\label{assumption.initialfidelityexponential}
    As discussed in Sec.~\ref{ssec.calibration_modeling}, the initial fidelity $F_e$ of the EPR pair generated at a certain time via link $e$ decreases approximately exponentially in time over activation period $a_e$ \cite{inestabuffer25} as:
    $$F_e =  \left( F_e^M - \frac{1}{4} \right) e^{- \Gamma a_e} + \frac{1}{4},$$ where $\Gamma$ is the exponential decay parameter and $F^M_e$ is the maximum initial fidelity that an EPR pair can be generated with, that is, when $a_e \to 0$.
\end{assumption}

\begin{assumption}\label{assumption.endtoendfidelity}
    The end-to-end fidelity $F_{ete}$ decreases exponentially with the number of links $N$ \cite{durpurification99}
\begin{align}
F_{ete} 
&= \frac{1}{4} \Bigg[ 1 
+ 3 \left( p_1^2 p_2 \right)^{N-1} 
\left( \frac{4 \eta^2 - 1}{3} \right)^{N - 1} \nonumber\\
&\quad \times \prod_{i=0}^{N-1} \left( \frac{4 F_i - 1}{3} \right) \Bigg],
\label{eq:end_to_end_fidelity}
\end{align}

where $F_i$ is the initial fidelity function (Assumption~\eqref{assumption.initialfidelityexponential}) of the $i^{th}$ link, $p_1$ is the efficiency of the single-qubit operation (Hadamard gate), $p_2$ is the efficiency of the two-qubit operation (CNOT gate), and $\eta$ is the efficiency of the Bell measurement for entanglement swapping operations.
\end{assumption}

\subsection{Constraints: Initial fidelity and end-to-end fidelity}

\begin{lemma}\label{lemma.initialfidelityvariation}
    For every link $e$, the activation period $a_e$ is bounded by the initial fidelity threshold $F^{th}_e$ set for the EPR pair generated in that link: 
    $$a_e \leq \frac{1}{\Gamma} \ln{\left( \frac{F_e^M - \frac{1}{4}}{F_e^{th} - \frac{1}{4}} \right) }.$$
\end{lemma}

\begin{proof}
Follows from Assumption~\ref{assumption.initialfidelityexponential} by setting  $F_e \geq F_e^{th}$.
\end{proof}
Lemma~\ref{lemma.initialfidelityvariation} gives the upper bound on the activation period of the link $e$. At the upper bound, we define
\begin{equation}
    a_e = a_e^{th} = \frac{1}{\Gamma} \ln{\left( \frac{F_e^M - \frac{1}{4}}{F_e^{th} - \frac{1}{4}} \right) }.
\end{equation}

The constraint on the fidelity of initially generated EPR pairs can be imposed to meet
the expectations of purification protocols, which require a certain fidelity of the initially
generated EPR pairs to perform purification on the links.

\begin{lemma}\label{lemma.endtoendfidelityactivation}
    The activation period $a_e$ is bounded by the end-to-end fidelity threshold $F_{ete}^{th}$ set for the EPR pair delivery along a path $$a_e \leq \frac{1}{\Gamma N} \ln{\left( 
 \frac{\left( \frac{4}{3} \right)^{N-1} \left( 
 F_e^M - \frac{1}{4}\right)^N \mathcal{U}}{F^{th}_{ete} - \frac{1}{4}}\right)},$$
 where $\mathcal{U} = \left( p_1^2 p_2 \right)^{N-1} \left( \frac{4 \ \eta^2 - 1}{3} \right)^{N - 1}$ defines the noisy operations during entanglement swapping.
\end{lemma}

\begin{proof}
Follows from Assumption~\ref{assumption.initialfidelityexponential} and \ref{assumption.endtoendfidelity} by setting $F_{ete} \geq F^{th}_{ete}$ and assuming $F_i=F_e.$
\end{proof}
Note that here, the end-to-end fidelity expression assumes identical noisy-operation parameters $(p_1,p_2,\eta)$ and identical functional form of $F_e$ across all links. Heterogeneity considerations can be addressed by simply extending the notation. The end-to-end fidelity of EPR delivery is a general metric used in the evaluation of the quality of delivery. A constraint on it refers to a certain threshold end-to-end fidelity requested
by users to run a quantum application on their end, or a certain threshold required by
purification protocols on end-to-end, similar to the initial fidelity constraint.

\subsection{Fidelity and Log-fidelity variable spaces}
In the upcoming analysis, it is often easier to work in the Log-fidelity variable space (L-space) than the linear fidelity variable space (F-space). Below is the relation between these two working spaces.
\begin{remark}
    From Assumption~\ref{assumption.endtoendfidelity}, we express initial fidelity in terms of the end-to-end fidelity threshold. Assume perfect operations, that is, $\mathcal{U}=1$ and same initial fidelity $F_i=F_e$ for each link we have
\begin{equation}
    F_e = \frac{1}{4}+\left( \frac{4}{3} \right)^{\left( \frac{1-N}{N} \right)} \left( F_{ete}^{th} - \frac{1}{4} \right)^{ \left( 1/N \right)}.
\end{equation}
Let $L_e = \ln{\left( \frac{A}{F_e - 1/4} \right)}$ where $A= F_e^M - \frac{1}{4}$ then
\begin{equation}
    \boxed{F_e = \frac{1}{4} + \frac{A}{e^{L_e}}}.
\end{equation}
\end{remark}

\section{Optimal calibration in linear quantum chains}\label{sec:calibration_prelim}
In this section, we delve into some preliminaries of the calibration-aware throughput optimisation problem. We discuss optimal allocation of the activation periods of links under different levels of constraints in linear quantum chains, that is, under the initial fidelity constraint in Section~\ref{ssec:initial_fidelity_const}, and under the end-to-end fidelity constraint in Section~\ref{ssec:ete_fidelity_cont}. This will then enable the discussion of the general link orchestration under constraints problem in Section~\ref{ssec:qlotheorem}.

\subsection{Linear quantum chains under initial fidelity constraint}\label{ssec:initial_fidelity_const}
\begin{theorem}[Optimal Throughput via Fidelity Threshold Selection]\label{theorem.optimalthroughput}
    Consider a quantum network chain comprising \( N \) links (that is, \( N+1 \) nodes) in USPC (Sec.~\ref{ssec.calibration_modeling}). Then, the throughput \( \mathcal{T} \) of the quantum network is maximised when each link is operated at its respective initial fidelity threshold
    $$F_e = F_e^{th} \ \text{or} \ L_e=L_e^{th}.$$ From Lemma~\ref{lemma.initialfidelityvariation}, in other words, the throughput $\mathcal{T}$ is maximized when the activation periods $a_e$ of links are operated at their respective thresholds. $$a_e = a_e^{th}.$$
\end{theorem}
\begin{proof}
    From Lemma~\ref{lemma.throughputN}, the throughput is strictly increasing in each activation duration $a_e$ (equivalently in each log-fidelity variable $L_e$) over the feasible region. Because the only restriction on $a_e$ is the upper bound imposed by the fidelity constraint, the maximum throughput is necessarily achieved at the largest feasible activation duration, namely $a_e^{th}$. No interior point can be optimal under a strictly increasing objective. The equation where the link fidelity allocation is selected by simply setting it to the value of the threshold is called the \textit{equation of optimality}.
\end{proof}

\subsection{Linear quantum chains under end-to-end fidelity constraint}\label{ssec:ete_fidelity_cont}
\begin{theorem}[Uniform Activation-Fidelity Theorem]\label{theorem.equalactivation}
Consider a quantum network chain comprising $N$ links in USPC. Then, the throughput $\mathcal{T}$ of the quantum network is maximised when all links are activated equally, and the fidelity of EPR generation in each link equals the optimal point $\Omega$ defined in Eq.~\eqref{eq:omega_f} and \eqref{eq:omega_l} as:
\begin{align}
\Omega_F(F_{ete}^{th}, N) = F_e 
&= \frac{1}{4}+\left( \frac{4}{3} \right)^{\left( \frac{1 - N}{N} \right)} 
\left( F_{ete}^{th} - \frac{1}{4} \right)^{1/N}, \nonumber \\ \text{(F-space)}
\label{eq:omega_f}
\end{align}

\begin{align}
\Omega_L(F_{ete}^{th}, N) = L_e 
&= \left( \frac{N - 1}{N} \right) \ln \left( \frac{4}{3} \right) 
+ \ln(A) \nonumber\\
&\quad - \frac{1}{N} \ln \left( F_{ete}^{th} - \frac{1}{4} \right).
\quad \text{(L-space)}
\label{eq:omega_l}
\end{align}
\end{theorem}
\begin{proof}
    Refer to Appendix~\ref{appendix.proofs}.
\end{proof}
Throughout this work, we will mainly be working in L-space; we will omit the subscript $L$ whenever required and simply denote as $\Omega$ the optimal point. $\Omega_{\pi}$ denotes the optimal point for a particular path $\pi$.

\subsection{Quantum Link Orchestration (QLO) theorem}\label{ssec:qlotheorem}
We now formally define the throughput optimisation problem accounting for the periodic calibration in a linear quantum chain as Problem~\ref{prob:problem1}. We then treat the problem analytically and provide its solution as Theorem~\ref{theorem.chain}.
\begin{problem}[Per-path calibration allocation]\label{prob:problem1}
For a given quantum network chain (also known as path $\pi$) in USPC, the throughput optimisation problem is to find the activation periods $a_e$
for each link satisfying the local individual initial fidelity threshold ($F_e^{th}$) and the end-to-end fidelity threshold ($F_{ete}^{th}$) such that the throughput ($\mathcal{T}_{\pi}$) of end-to-end entanglement across the path $\pi$ is maximised. We work in $L$-space with $L_e=\Gamma a_e$ and $K=\Gamma c$, so that
$\frac{a_e}{a_e+c}=\frac{L_e}{L_e+K}$, end-to-end log-fidelity budget
$\mathcal{L}=(N-1)\ln\!\frac{4}{3}+N\ln A-\ln\!\bigl(F_{ete}^{th}-\tfrac14\bigr)$,
and per-link log-fidelity caps $L_e^{th}$, and find the log-fidelity allocation
$\{L_e\}_{e\in\pi}$ solving
\begin{align}
\max_{\{L_e\}_{e\in\pi}}\quad
& \mathcal{T}_\pi \;=\; \mathcal{T}_{nc} \prod_{e\in\pi}\frac{L_e}{L_e+K} \nonumber\\
\text{s.t.}\quad
& \sum_{e\in\pi} L_e \;\le\; \mathcal{L}, \label{eq:P1-budget}\\
& 0 \;\le\; L_e \;\le\; L_e^{th}, \qquad \forall\, e\in\pi, \label{eq:P1-cap}
\end{align}
where activation periods are recovered as
$a_e = L_e/\Gamma$.
\end{problem}
\begin{remark}[Separable concave structure]\label{remark.concave}
Since $K=\Gamma c$ is fixed, $\ln\mathcal{T}_\pi = \ln \mathcal{T}_{nc} + \sum_{e\in\pi}\ln\frac{L_e}{L_e+K}$
is a sum of separable terms $\phi(L_e)=\ln\frac{L_e}{L_e+K}$, each strictly
increasing and strictly concave on $(0,\infty)$ ($\phi''(L_e)=\frac{1}{(L_e+K)^2}-\frac{1}{L_e^2}<0$).
Problem~\ref{prob:problem1} is therefore the maximisation of a separable concave
objective over a convex polytope, solved optimally by Theorem~\ref{theorem.chain}.
\end{remark}

\begin{theorem}[Quantum Link Orchestration Theorem]\label{theorem.chain}
    Consider a quantum network chain (also known as path $\pi$) having optimal point $\Omega_{\pi}$ comprising $N$ links in UPSC. Let $L_i^{th}$ denote the threshold fidelity for initial entanglement generation on the $i^{th}$ link. Then, the throughput $\mathcal{T}$ of the quantum network is maximised by allocating links $L_i$ as follows
    \begin{enumerate}
        \item For $ L_i^{th} \leq \Omega_{\pi}: L_i = L_i^{th}$
        \item For $\Omega_{\pi} < L_i^{th}: L_i = \Omega_{\pi}$
        \item For $L_j^{th} \leq \Omega_{\pi} \ \& \ \Omega_{\pi} < L_k^{th}$
        \begin{enumerate}
            \item if $\sum L_j^{th} + \sum L_k^{th} \leq \mathcal{L}: L_j = L_j^{th} \ \& \ L_k = L_k^{th}$
            \item if $\sum L_j^{th} + \sum L_k^{th} > \mathcal{L}: L_j = L_j^{th}$ \ \& \ repeat the procedure for $L_k $ links with their corresponding initial fidelity thresholds $L_k^{th} $ and optimal point for the path $\Omega_{\pi}^{\gamma} = \frac{1}{k} \left[ (j+k) \; \Omega_{\pi} - \left( \sum L_j^{th}\right) \right]$. Here $\gamma \in \mathbb{N}$ denotes the level of optimisation.
        \end{enumerate}
    \end{enumerate}
    where $\sum_{i=0}^{N-1} L_i$ $\leq$ $\mathcal{L} = (N-1) \ln{\left( \frac{4}{3} \right)} + N \ln{A} - \ln{ \left( F_{ete}^{th} - \frac{1}{4} \right)} $
\end{theorem}
\begin{proof}
    According to Theorem~\ref{theorem.optimalthroughput}, $L_i = L_i^{th}$. In addition, from Theorem~\ref{theorem.equalactivation} (Eq.~\eqref{eq:sumLi}), the end-to-end fidelity defines the constraint that is, $\sum_{i=0}^{N-1} L_i \leq (N-1) \ln{\left( \frac{4}{3} \right)} + N \ln{A} - \ln{ \left( F_{ete}^{th} - \frac{1}{4} \right)}$. Let the individual thresholds on the initial fidelity be $L_i \leq L_i^{th}$. 
    Essentially, the subsection of N-dim L-hyperspace under consideration is defined by
    \begin{enumerate}
        \item Equation of optimality: $L_i = L_i^{th}$
        \item End-to-end fidelity constraint: $\sum_{i=0}^{N-1} L_i \leq (N-1) \ln{\left( \frac{4}{3} \right)} + N \ln{A} - \ln{ \left( F_{ete}^{th} - \frac{1}{4} \right)}$
        \item Individual initial fidelity threshold constraint: $L_i \leq L_i^{th}$
    \end{enumerate}
    Consider the following cases:
    \begin{enumerate}
        \item For $L_i^{th} \leq \Omega_{\pi}$, all the individual initial fidelity thresholds are less than the optimal point $\Omega_{\pi}$. In such a case, constraints are fulfilled trivially; hence, using the equation of optimality, the throughput is maximised for $L_i=L_i^{th}.$
        \item For $\Omega_{\pi}<L_i^{th}$, all the individual initial fidelity thresholds are more than the optimal point $\Omega_{\pi}.$ In such a case, operating at optimal using the equation of optimality would violate constraint 2. Hence, to satisfy constraint 2, a simple optimal point is used, that is, the throughput is maximised when $L_i=\Omega_{\pi}.$
        \item For $L_j^{th} \leq \Omega_{\pi} \ \& \ \Omega_{\pi} < L_k^{th}$, that is, $j$ of the individual initial fidelity thresholds are less than the optimal point $\Omega_{\pi}$, and $k$ of the individual initial fidelity thresholds are more than the optimal point $\Omega_{\pi}$. In such a case, constraint 2 can be satisfied or violated depending upon the individual initial fidelity thresholds. 
        \begin{enumerate}
            \item If constraint 2 is satisfied that is, $\sum L_j^{th} + \sum L_k^{th} \leq \mathcal{L} $ where $\mathcal{L} = \sum_{i=0}^{N-1} L_i = (N-1) \ln{\left( \frac{4}{3} \right)} + N \ln{A} - \ln{ \left( F_{ete}^{th} - \frac{1}{4} \right)} $, then similar to case 1, using the equation of optimality, the throughput is maximized for $L_j=L_j^{th} \ \& \ L_k=L_k^{th}.$
            \item If the constraint 2 is violated, that is, $\sum L_j^{th} + \sum L_k^{th} > \mathcal{L}$ where $\mathcal{L} = \sum_{i=0}^{N-1} L_i = (N-1) \ln{\left( \frac{4}{3} \right)} + N \ln{A} - \ln{ \left( F_{ete}^{th} - \frac{1}{4} \right)} $ then, $j$ of the links are operated using the equation of optimality that is, $L_j=L_j^{th}$ and hence reducing the original $N$-dim $L$-hyper-space to $(N-j)$-dim $L$-hyper-space. The remaining $k$ of the links are evaluated again with an increased optimal point given by $\Omega_{\pi}^{1} = \frac{1}{k} \left[ (j+k) \Omega_{\pi} - \left( \sum L_j^{th}\right) \right]$ following the same procedure. See Appendix~\ref{appendix:optimal_point_proof} for proof of $\Omega^{\gamma+1}_{\pi} \geq \Omega^{\gamma}_{\pi}$. The recursion follows until all the links are optimally allocated at the level of optimisation $\gamma$ with optimal point $\Omega_{\pi}^{\gamma+1} = \frac{1}{k} \left[ 
 (j+k) \; \Omega_{\pi}^{\gamma} - \sum_j L_j^{th} \right].$
        \end{enumerate}
    \end{enumerate}
\end{proof}

Algorithm~\ref{alg:QLO} implements Theorem~\ref{theorem.chain} constructively.
\textbf{Lines~1--7} initialise the optimisation level $\gamma=0$, set the current
optimal point $\Omega\gets\Omega_\pi$, and define the fixed set $\mathcal{F}=\emptyset$
and remaining set $\mathcal{R}=\{1,\ldots,N\}$ with all allocations zeroed.
\textbf{Line~9} identifies $\mathcal{S}\subseteq\mathcal{R}$, the subset of links
whose thresholds already lie at or below $\Omega$; if $\mathcal{S}=\emptyset$
(\textbf{lines~10--15}), all remaining links are set to $\Omega$, and the algorithm
returns immediately.
\textbf{Lines~16--17} fix every link $i\in\mathcal{S}$ at its threshold
$L_i\gets L_i^{th}$, moving it from $\mathcal{R}$ to $\mathcal{F}$; if
$\mathcal{R}$ is now empty (\textbf{lines~21--23}), the algorithm terminates.
\textbf{Lines~24--30} perform a feasibility check: if the fixed allocations
plus the thresholds of all remaining links still satisfy the end-to-end budget
$S_F + S_R^{th}\le\mathcal{L}$, the remaining links are set to their thresholds
and the algorithm returns.
Otherwise (\textbf{lines~32--35}), the budget would be violated, so the optimal
point is updated as
$\Omega\gets\frac{1}{k}\bigl((j+k)\Omega_\pi - S_F\bigr)$,
the optimisation level is incremented, and the loop repeats on the reduced
set $\mathcal{R}$. A worked two-link example illustrating Cases~1--3 of Theorem~\ref{theorem.chain},
including the throughput heat maps in F- and L-space, is provided in
Appendix~\ref{sec:two_link_example}.

\begin{algorithm}[t]
\caption{Quantum Link Orchestration for a Path $\pi$}
\label{alg:QLO}
\begin{algorithmic}[1]
\REQUIRE Path $\pi$ with $N$ links, initial fidelity thresholds $\{L_i^{th}\}_{i=1}^N$,
end-to-end budget $\mathcal{L}$ and base optimal point $\Omega_\pi$ both derived from end-to-end fidelity threshold $F_{ete}^{th}$
\ENSURE Allocated link values $\{L_i\}_{i=1}^N$

\STATE Initialise optimisation level $\gamma \gets 0$
\STATE Set current optimal point $\Omega \gets \Omega_\pi^{\gamma}$
\STATE Initialise fixed set $\mathcal{F} \gets \emptyset$
\STATE Initialise remaining set $\mathcal{R} \gets \{1,2,\ldots,N\}$
\FOR{$i=1$ \TO $N$}
    \STATE $L_i \gets 0$
\ENDFOR

\WHILE{$\mathcal{R} \neq \emptyset$}

    \STATE Identify $\mathcal{S} \gets \{ i \in \mathcal{R} \mid L_i^{th} \le \Omega \}$

    \IF{$\mathcal{S} = \emptyset$}
        \FOR{$i \in \mathcal{R}$}
            \STATE $L_i \gets \Omega$
        \ENDFOR
        \RETURN $\{L_i\}_{i=1}^N$
    \ENDIF
    
    \FOR{$i \in \mathcal{S}$}
        \STATE $L_i \gets L_i^{th}$
        \STATE $\mathcal{F} \gets \mathcal{F} \cup \{i\}$
        \STATE $\mathcal{R} \gets \mathcal{R} \setminus \{i\}$
    \ENDFOR

    \IF{$\mathcal{R} = \emptyset$}
        \RETURN $\{L_i\}_{i=1}^N$
    \ENDIF

    \STATE Compute $S_F \gets \sum_{i \in \mathcal{F}} L_i$
    \STATE Compute $S_R^{th} \gets \sum_{i \in \mathcal{R}} L_i^{th}$

    \IF{$S_F + S_R^{th} \le \mathcal{L}$}
        \FOR{$i \in \mathcal{R}$}
            \STATE $L_i \gets L_i^{th}$
        \ENDFOR
        \RETURN $\{L_i\}_{i=1}^N$
    \ELSE
        \STATE $j \gets |\mathcal{F}|$ \hspace{1em} $k \gets |\mathcal{R}|$
        \STATE Update optimal point:
        \STATE $\Omega \gets \frac{1}{k}\Big( (j+k)\Omega_\pi^{\gamma} - S_F \Big)$
        \STATE $\gamma \gets \gamma + 1$
    \ENDIF

\ENDWHILE

\end{algorithmic}
\end{algorithm}

\subsection{Complexity and optimization depth}\label{ssec:complexity}
\begin{corollary}[Worst-Case Complexity]
\label{cor:qlo_complexity}
In the worst case, Algorithm~\ref{alg:QLO} runs in $O(N^2)$ time.
\end{corollary}
\begin{proof}
At each iteration, the algorithm scans the remaining set $\mathcal{R}$ to identify
$\mathcal{S}=\{i\in\mathcal{R}\mid L_i^{th}\le \Omega\}$, which costs $O(|\mathcal{R}|)$.
In the worst case, exactly one link is fixed per iteration, so $|\mathcal{R}|$ decreases
from $N$ to $1$. Hence, the total scanning cost is
$\sum_{m=1}^{N} O(m)= O\left( \frac{N(N+1)}{2}\right)= O(N^2)$.
All additional per-iteration updates and sum computations are at most linear in $N$,
therefore the overall worst-case time complexity is $O(N^2)$.
\end{proof}

\begin{remark}[Optimisation Level]
\label{remark.optlevel}
Algorithm~\ref{alg:QLO} proceeds through a sequence of \emph{optimisation levels} indexed by
$\gamma \in \mathbb{N}$, where each level corresponds to an update of the effective optimal point
$\Omega_\pi^\gamma$, which is used to allocate the links that remain unresolved.
\end{remark}
Since at least one link is allocated on each optimisation level, the process terminates in a finite
number of levels, and the final allocation assigns each link either to its threshold $L_i^{th}$ or according to the
appropriate updated optimal point $\Omega_\pi^\gamma$.

\section{Calibration aware orchestration in general networks}\label{sec:calibration_general}
Having established the optimal solution for calibration-aware orchestration in a linear quantum network chain, we now explore the problem in a general network. 
\begin{problem}[General-network calibration allocation]\label{prob:P2}
Let a general quantum network that carries a set of fixed source-destination paths
$P=\{\pi_1,\dots,\pi_{|P|}\}$ (one per SD pair, obtained by a given routing
policy). For each link $e$ let $K=\Gamma c$ and per-link threshold $L_e^{th}$,
and for each path $\pi\in P$ let $\mathcal{L}_\pi$ be its end-to-end log-fidelity threshold. Let $\Pi_e=\{\pi\in P : e\in\pi\}$ denote the paths traversing link $e$.
A single physical activation period is allocated per link, i.e. one variable
$L_e$ shared by all $\pi\in\Pi_e$. Then find $\{L_e\}$ solving
\begin{align}
\max_{\{L_e\}_{e\in E}}\quad
& T = \sum \mathcal{T}_{\pi}=\sum_{\pi\in P} {\mathcal{T}_{nc}}^{\pi} \prod_{e\in\pi}\frac{L_e}{L_e+K} \label{eq:P2-obj}\\
\text{s.t.}\quad
& \sum_{e\in\pi} L_e \;\le\; \mathcal{L}_\pi, \qquad \forall\, \pi\in P, \label{eq:P2-budget}\\
& 0 \;\le\; L_e \;\le\; L_e^{th}. \label{eq:P2-cap}
\end{align}
The objective in Eq.~\eqref{eq:P2-obj} sums only the throughput of paths meeting their
own end-to-end fidelity threshold $F_{ete,\pi}^{th}$. The activation periods are
recovered as $a_e = L_e/\Gamma$.
\end{problem}
\begin{remark}[Loss of separable concave structure]\label{remark.network_nonconvex}
Unlike the per-path objective of Problem~\ref{prob:problem1}, the network
objective in Eq.~\eqref{eq:P2-obj} is a \emph{sum} of the per-path throughputs
$\mathcal{T}_{\pi} = \mathcal{T}_{nc}^{\pi}\prod_{e\in\pi}\frac{L_e}{L_e+K}$. Each $\mathcal{T}_\pi$
is log-concave (Remark~\ref{remark.concave}), but log-concavity is not preserved
under addition, and a shared link couples its paths through a common variable
$L_e$. The resulting objective $T$ is therefore neither separable nor jointly
concave, so the Problem~\ref{prob:P2} is a non-convex program for which no
closed-form per-link optimum analogous to Theorem~\ref{theorem.chain} is available.
When no links are shared, $T$ decouples into the disjoint per-path problems of
Problem~\ref{prob:problem1}, each solved optimally by QLO. With shared links, we
resort to two feasible strategies: the conservative MIN/GRO heuristics
(Section~\ref{ssec:minimum_selection_rule},~\ref{ssec:greedy_qlo}) and the
numerical optimiser NUM (Section~\ref{ssec:approaches}), which attains only a
local optimum of this non-convex problem.
\end{remark}
We first extend QLO to a general quantum network by utilising the \emph{minimum selection rule} (MIN), which serves as a baseline heuristic (Section~\ref{ssec:minimum_selection_rule}). We then propose a greedy quantum link orchestration (GRO) heuristic for the calibration problem in general quantum networks (Section~\ref{ssec:greedy_qlo}).
\begin{remark}[On the notion of a shared link]
    In the quantum networking literature, the term \emph{link} is used in two distinct senses: 
    (i)~the \emph{physical link}, comprising a quantum channel and/or a classical control 
    channel connecting neighbouring nodes, and (ii)~the \emph{entanglement pair} distributed 
    between those nodes in the abstract, protocol-level sense. The physical link is a reusable 
    resource since it can generate entanglement pairs repeatedly across successive activation 
    periods, whereas a distributed entanglement pair is consumed upon use, e.g., during 
    entanglement swapping to extend entanglement over longer distances. Throughout this section, 
    \emph{shared link} refers exclusively to the physical link in sense~(i).

    Shared-link paths naturally arise when the activation period of a particular link spans multiple time slots for a fixed set of source and destination pairs utilising end-to-end entanglement. In such a scenario, while two paths within the same time slot may not share links, shared links can still emerge across paths belonging to different time slots, since routing paths are recalculated after each time slot as depicted in Fig.~\ref{fig:calibration_timeslot}. The red blocks $\epsilon_i$ are individual entangled-pair deliveries on the end-to-end timeline. The orange blocks $\theta_0, ...,\,\theta_3 $ are the \textit{network level request processing time slots} where for first block $\theta_0$ a certain number of entangled pair are delivered between source-destination pair members ({(Alice 1, Bob 2), (Alice 2, Bob 1)}) -- same color for same source-destination pair until one (or more) members have met its required number of entangled pair under current request or admission of incoming request under a given protocol thereby ending processing time slot. For the next slot, again the same procedure follows, which allows additional source-destination members with their own entangled pair requirements. The activation period ($a_e^0$) (in green) and calibration period ($c_e^0$) (in blue) are representations of a general link whose activation period is involved in both time slots $\theta_0$ and $\theta_1$. However, each link's activation and calibration period could overlap with different time slots as they are operated independently.
\begin{figure*}[!t]
\centering
\includegraphics[width=2\columnwidth]{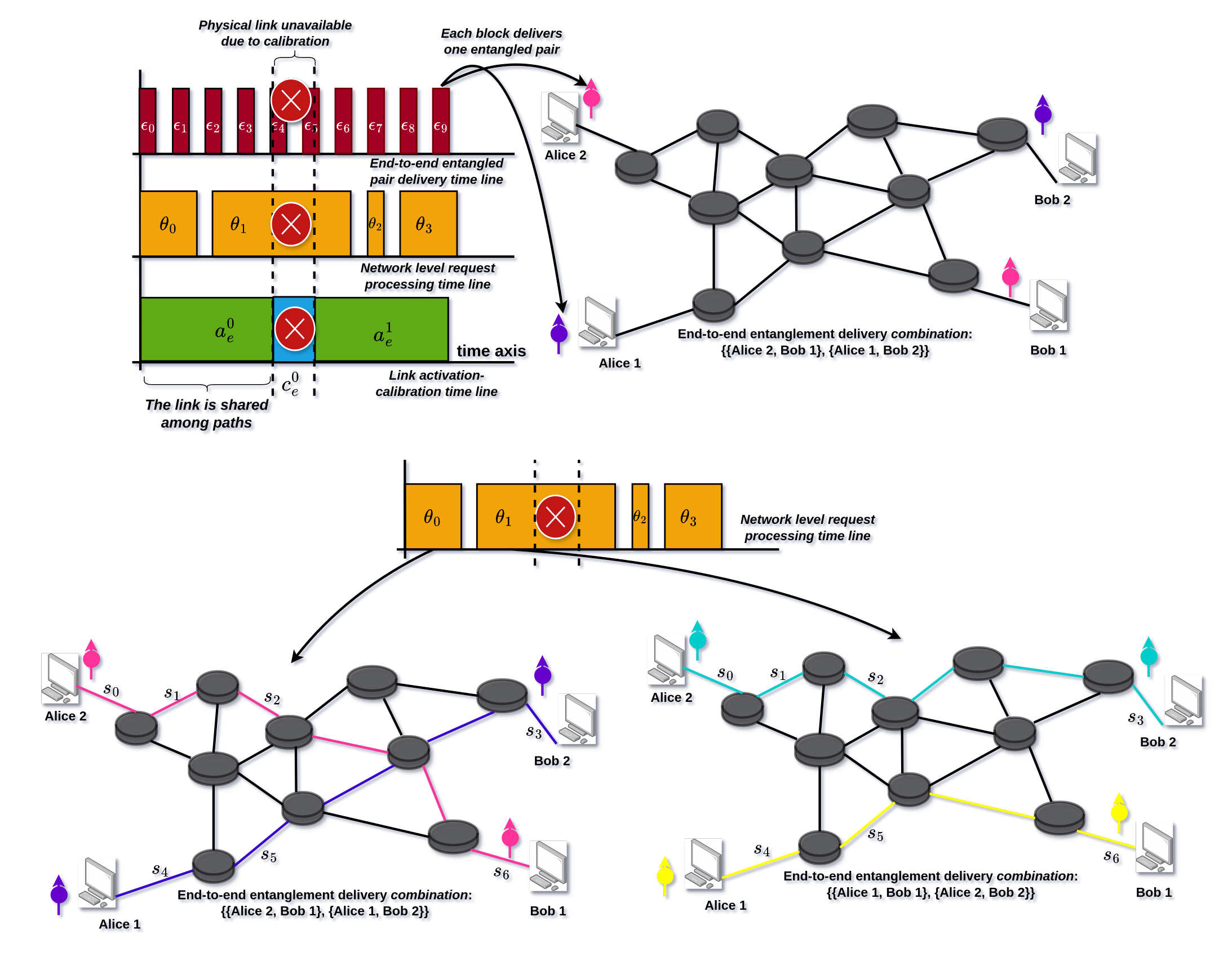}
\caption{\textbf{Timing diagram.} (Top) Each red block represents the delivery of a single entangled pair between a dedicated source–destination pair in a combination. Each orange block denotes the processing time associated with a request for delivering a specific number of entangled pairs between given source–destination pairs. The green and blue blocks correspond to the activation and calibration periods, respectively, of a general physical link $e$. The overlap of the activation period $a_e^0$ with the time slots of requests $\theta_0$ and $\theta_1$ induces link sharing among the routing paths involved in the delivery of entangled pairs, particularly for the blocks $\epsilon_i$, $\forall i \in \{0,1,2,3\}$. (Bottom) The connections switch from \{\{Alice 2, Bob 1\}, \{Alice 1, Bob 2\}\} in time slot $\theta_0$ to \{\{Alice 1, Bob 1\}, \{Alice 2, Bob 2\}\} in time slot $\theta_1$, inducing overlap on the shared links $s_0, s_1, \dots, s_6$.
}
\label{fig:calibration_timeslot}
\end{figure*}
\end{remark}
\subsection{Minimum selection rule benchmark}\label{ssec:minimum_selection_rule}

\begin{proposition}[General overlapping-link calibration via minimum selection rule]
\label{prop:general-overlap}
Consider a set of $P$ paths $\{\pi_1, \ldots, \pi_P\}$ sharing 
a common quantum network graph, obtained by routing $P$ 
source-destination pairs via a general routing policy. By 
Theorem~\ref{theorem.chain}, each path $\pi_k$ admits an individually optimal 
QLO allocation $L_{\{\pi_k, e\}}$ for every link $e \in \pi_k$.

For any link $e$, let $\Pi_e \subseteq \{\pi_1,\ldots,\pi_P\}$ denote 
the set of paths that traverse $e$. A feasible benchmark allocation is 
obtained by assigning
\[
    L_e = \min_{\pi \in \Pi_e}\; L_{\{\pi, e\}},
    \qquad \forall\, e \ \text{ with } \Pi_e \neq \varnothing.
\]
Links not traversed by any path incur no overlap constraint.
\end{proposition}
Algorithm~\ref{alg:multi-path} (Appendix~\ref{appendix.minimum_rule}) implements this in a two-pass structure, running in $O(M)$ time where $M \;=\; \sum_{\pi} |\pi|$; see Appendix~\ref{appendix.minimum_rule} for the pseudocode and complexity proof.

\begin{remark}[Overall complexity with overlapping paths]
\label{remark:overall-complexity}
When extending the linear-chain orchestration of Theorem~\ref{theorem.chain}
to a general quantum network with $P$ source--destination paths, the overall
procedure decomposes into two stages.

\begin{enumerate}
    \item \textbf{Per-path orchestration (Stage A).}
    Algorithm~\ref{alg:QLO} is applied independently to each of the $P$
    paths. Let $N_\pi$ denote an upper bound on the number of links in any
    path. By Corollary~\ref{cor:qlo_complexity}, the worst-case cost per
    path is $O(N_\pi^2)$, giving a total Stage A cost of
    \[
        O\bigl(P \cdot N_\pi^2\bigr).
    \]

    \item \textbf{Shared-link resolution (Stage B).}
    The per-path QLO allocations are harmonised using
    Algorithm~\ref{alg:multi-path}. By
    Corollary~\ref{cor:multi-path-complexity}, Stage B costs $O(M)$.
\end{enumerate}

Combining both stages, the overall running time is
\[
    O\bigl(P \cdot N_\pi^2 + M\bigr).
\]
\end{remark}

\subsection{Greedy quantum link orchestration}
\label{ssec:greedy_qlo}
 
While the minimum selection rule provides a simple and conservative benchmark for
resolving overlapping-link constraints, it does not exploit the calibration budget
freed on each path when shared links are allocated below their path-optimal values.
To address this limitation, we propose \emph{Greedy Quantum Link Orchestration}
(GRO), a heuristic that redistributes the freed budget locally on each path while
preserving feasibility across all paths. Like the minimum selection rule, GRO is
applied independently to each SD-pair combination involving fixed paths between each source-destination pair of combination.
 
\paragraph{Two-pass structure.}
GRO operates in two sequential passes over the network.
 
\emph{Pass~1 (conflict resolution).}
For every link $e$, GRO applies the minimum selection rule across all paths
in the combination that shares it:
\[
    L_e \;=\;
    \begin{cases}
        \min_{\pi \in \Pi_e}\; L_{\{\pi,e\}} & \text{if } |\Pi_e| \ge 2,\\
        L_{\{\pi,e\}}                         & \text{if } \Pi_e = \{\pi\},
    \end{cases}
\]

All shared links are resolved in one flat pass with no priority ordering among them.
This pass is equivalent to Algorithm~\ref{alg:multi-path} and guarantees feasibility
for every path.
 
\emph{Pass~2 (greedy reallocation).}
After Pass~1, each path $\pi$ may have unused budget. For each path independently,
GRO computes the residual
\[
    \Delta_\pi \;=\; \mathcal{L}_\pi - \sum_{e \in \pi} L_e \;\ge\; 0,
\]
and redistributes it over the exclusive (unshared) links of $\pi$:
\[
    \pi_{\mathrm{un}} \;=\; \bigl\{\, e \in \pi \mid |\Pi_e| = 1 \,\bigr\}.
\]
Among the links in $\pi_{\mathrm{un}}$ with strictly positive available headroom,
links are sorted in \emph{descending order of headroom}
\[
    r_e \;=\; L_e^{th} - L_e \;>\; 0,
\]
Budget $\Delta_\pi$ is allocated greedily in this order: each link $e$ receives
$\delta = \min(r_e,\, \Delta_\pi)$, and $\Delta_\pi$ is decremented accordingly
until it is exhausted or all unshared links reach their caps.
Because unshared links are exclusive to $\pi$, this reallocation cannot affect any
other path, and the end-to-end constraint $\sum_{e \in \pi} L_e \le \mathcal{L}_\pi$
is maintained throughout.
The pseudocode is given as Algorithm~\ref{alg:gro_hybrid}.
 
\begin{proposition}[Feasibility and improvement over minimum selection rule]
\label{prop:gro-feasibility}
Algorithm~\ref{alg:gro_hybrid} produces an allocation $\{L_e\}$ that:
\begin{enumerate}
    \item is feasible, i.e., $\sum_{e \in \pi} L_e \le \mathcal{L}_\pi$ for every
          path $\pi \in P$; and
    \item weakly dominates the minimum selection rule, i.e.,
          $\mathcal{T}_\pi^{\mathrm{GRO}} \ge \mathcal{T}_\pi^{\mathrm{MIN}}$ for every path $\pi$,
          with strict inequality whenever $\Delta_\pi > 0$ and at least one
          unshared link of $\pi$ has positive headroom $r_e > 0$.
\end{enumerate}
\end{proposition}
 
\begin{proof}
\begin{enumerate}
\item \textit{Feasibility.}
Pass~1 assigns $L_e \le L_{\{\pi,e\}}$ on shared links and
$L_e = L_{\{\pi,e\}}$ on exclusive links, so
$\sum_{e \in \pi} L_e \le \mathcal{L}_\pi$ after Pass~1, since the QLO
allocations already satisfy the end-to-end budget.
Pass~2 adds increments $\delta \le r_e$ to exclusive links only, and the greedy
loop terminates as soon as $\Delta_\pi$ is exhausted, maintaining
$\sum_{e \in \pi} L_e \le \mathcal{L}_\pi$ throughout.
 
\item \textit{Improvement.}
Pass~1 is identical for both GRO and MIN.
Pass~2 adds non-negative increments to exclusive links, strictly increasing their
$L_e$ values whenever $\Delta_\pi > 0$ and some $r_e > 0$.
Since $\mathcal{T}_\pi = \mathcal{T}_{nc} \prod_{e \in \pi} (a_e/(a_e + c))$ is strictly increasing in
each $a_e = L_e / \Gamma$ (Lemma~\ref{lemma.throughputN}), any such increment
yields $\mathcal{T}_\pi^{\mathrm{GRO}} > \mathcal{T}_\pi^{\mathrm{MIN}}$.
If $\Delta_\pi = 0$ or all unshared links are at their maximum caps, the
allocations coincide and $\mathcal{T}_\pi^{\mathrm{GRO}} = \mathcal{T}_\pi^{\mathrm{MIN}}$.
\end{enumerate}
\end{proof}

\paragraph{Complexity.}
Pass~1 scans all path-link incidences once and maintains a running minimum per
shared link: $O(M)$ with standard bookkeeping.
Pass~2 sorts the exclusive links of each path $\pi$ with positive headroom by
descending $r_e$ and performs a linear scan of those links: $O(k_\pi \log k_\pi)$
per path, where $k_\pi = |\pi_{\mathrm{un}}|$. Summing over all paths gives
$O(M \log M)$ in the worst case, where the sort dominates.
Combined with the $O(P \cdot N_\pi^2)$ cost of the preceding per-path QLO stage,
the end-to-end complexity of GRO is
\[
    O\bigl(P \cdot N_\pi^2 + M \log M\bigr).
\]
 
\begin{algorithm}[!t]
\caption{Greedy Quantum Link Orchestration (GRO) for shared-link resolution}
\label{alg:gro_hybrid}
\begin{algorithmic}[1]
\REQUIRE Set of paths $P$;
per-path QLO allocations $L_{\{\pi,e\}}$ for each $\pi \in P$, $e \in \pi$;
graph $G$ with per-link fidelity thresholds $F_e^{th}$ (used to compute $L_e^{th}$);
path budgets $\mathcal{L}_{\pi}$ derived from end-to-end fidelity threshold path $F_{ete, \ \pi}^{th}$
\ENSURE Updated per-link assignment $\{L_e\}$ and corresponding per-path allocations
 
\STATE \textit{// Pass 1: resolve shared links via minimum selection}
\FORALL{$e$ used by at least one path in $P$}
  \STATE $\Pi_e \gets \{\pi \in P \mid e \in \pi\}$
  \IF{$|\Pi_e| \ge 2$}
    \STATE $L_e \gets \min_{\pi \in \Pi_e}\; L_{\{\pi,e\}}$
  \ELSE
    \STATE $L_e \gets L_{\{\pi,e\}}$ \quad ($\Pi_e = \{\pi\}$)
  \ENDIF
\ENDFOR
 
\STATE \textit{// Pass 2: greedily redistribute freed budget over exclusive links}
\FORALL{$\pi \in P$}
  \STATE $\Delta_{\pi} \gets \mathcal{L}_{\pi} - \sum_{e \in \pi} L_e$
  \IF{$\Delta_{\pi} > 0$}
    \STATE $U_{\pi} \gets \{e \in \pi \mid |\Pi_e| = 1\}$ \quad (exclusive links of $\pi$)
    \STATE Compute $r_e \gets L_e^{th} - L_e$ for each $e \in U_{\pi}$
    \STATE Let $U_{\pi}^+ \gets \{e \in U_{\pi} \mid r_e > 0\}$
    \STATE Sort $U_{\pi}^+$ by descending $r_e$
    \FORALL{$e \in U_{\pi}^+$ in sorted order}
      \IF{$\Delta_{\pi} \le 0$}
        \STATE \textbf{break}
      \ENDIF
      \STATE $\delta \gets \min(r_e,\;\Delta_{\pi})$
      \STATE $L_e \gets L_e + \delta$;\quad $\Delta_{\pi} \gets \Delta_{\pi} - \delta$
    \ENDFOR
  \ENDIF
\ENDFOR
\RETURN $\{L_e\}$ and the corresponding per-path activation periods $a_e = L_e / \Gamma$
\end{algorithmic}
\end{algorithm}

\section{Simulations and Performance Evaluation}\label{sec:simulation}

\begin{table}[t]
\centering
\caption{Comparison of calibration approaches based on their underlying assumptions. REF and QLO represent upper bounds and are inherently infeasible by construction, whereas MIN, GRO, and NUM are feasible strategies applicable to general quantum networks.}
\label{tab:comparison}
\setlength{\tabcolsep}{3pt}
\renewcommand{\arraystretch}{1.1}

\begin{tabular}{>{\centering\arraybackslash}p{0.13\columnwidth}
                >{\centering\arraybackslash}p{0.23\columnwidth}
                >{\centering\arraybackslash}p{0.23\columnwidth}
                >{\centering\arraybackslash}p{0.24\columnwidth}}
\hline
\textbf{Method} &
\makecell[c]{\textbf{Local thresholds} \\ $(F^{th}_e)$} &
\makecell[c]{\textbf{E2E thresholds} \\ $(F^{th}_{ete})$} &
\makecell[c]{\textbf{Shared link} \\ \textbf{treatment}} \\
\hline
\textbf{REF} & $\checkmark$ & $\times$ & $\times$ \\
\textbf{QLO} & $\checkmark$ & $\checkmark$ & $\times$ \\
\textbf{MIN/GRO} & $\checkmark$ & $\checkmark$ & $\checkmark$ \\
\textbf{NUM} & $\checkmark$ & $\checkmark$ & $\checkmark$ \\
\hline
\end{tabular}
\end{table}

To complement the theoretical study till now, we developed in Python a specialised simulator QNetCal\footnote{The simulator will be released to the public with the submission of the camera-ready version of this paper.} for this work to compare the performances of different approaches, including relevant reference (REF), upper bound (QLO) and numerical approach (NUM). We consider a grid-based quantum network in which $n$ source--destination (SD)
pairs are selected uniformly at random over an $n\times n$ grid graph.
Each SD pair is assigned an independent end-to-end fidelity threshold
$F_{ete,\pi}^{th}$ drawn uniformly from a specified fidelity range.
For each value of $n \in \{1, 2, \ldots, 25\}$, we generate 100 independent
random SD-pair combinations and evaluate all approaches on each combination.
Results are reported as averages over these 100 combinations.
The physical link parameters are set to $\Gamma = 0.001$ and
$c = 0.13$ (expressed in the same arbitrary time unit), i.e.\ $K=\Gamma c = 1.3\times10^{-4}$,
with $F_e^M = 0.999$ and $\mathcal{A} = 1000$ throughout. Since fidelity depends on
$(\Gamma,c)$ only through $L_e=\Gamma a_e$ and throughput only through
$L_e/(L_e+K)$, $K$ is the single calibration parameter of the simulation. We selected its value by evaluating different parameter combinations and retaining a representative setting that provides clear and readable results. The choice of $K$ rescales the throughput gaps but, as shown
below, does not change the ordering of the approaches, and it leaves the structure of
delivered fidelities of REF, QLO, MIN and GRO unchanged. Additionally, the individual link fidelity threshold $F_e^{th}$ is selected from a uniform distribution over $[0.85, 0.99]$, allowing for a mix of higher and lower local fidelity thresholds.
Routing follows Dijkstra's shortest-path selection between each SD pair.
All simulations were carried out on a regular server. Wall-clock runtime was recorded for each approach: NUM required on the order of \emph{hours} per
fidelity threshold range, whereas GRO (and MIN) completed in the order of \emph{minutes}.

The primary metric is the sum of path throughputs ($T$) (not to be confused with $\mathcal{T}$, which is for a single path throughput) over all paths in a combination that satisfies their individual
end-to-end fidelity threshold $F_{ete,\pi}^{th}$, averaged over the 100
random combinations at each $n$.
Paths that deliver fidelity below their threshold are excluded from the sum
to reflect realistic operational requirements.
The per-SD-pair throughput is additionally reported by normalising the sum of path throughputs by $n$.

\begin{figure*}[!t]
\centering
\includegraphics[width=2\columnwidth]{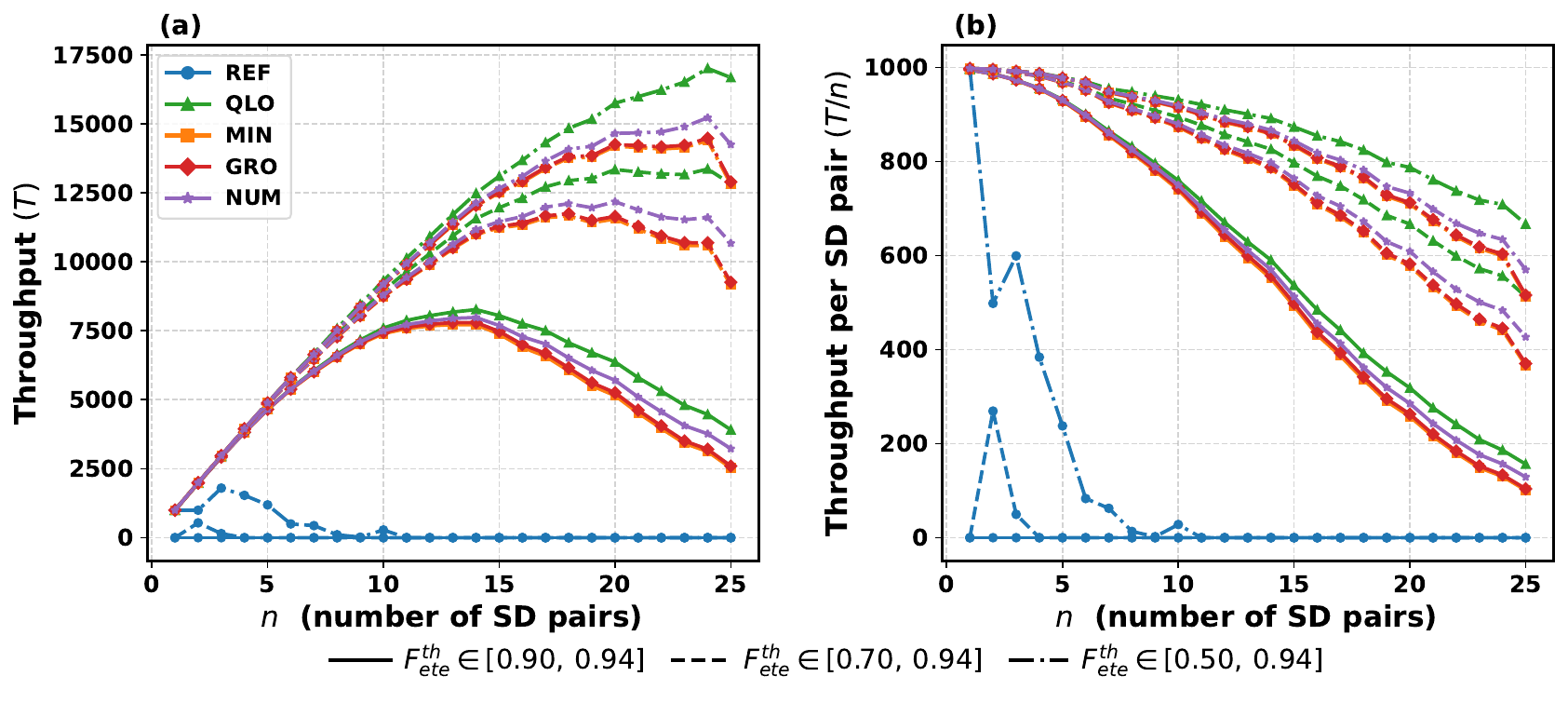}
\caption{Throughput as a function of the number of SD pairs $n$
    for all five calibration approaches across three end-to-end fidelity thresholds
    ranges, $F_{ete}^{th} \in [0.90, 0.94]$, $F_{ete}^{th} \in [0.70, 0.94]$, and $F_{ete}^{th} \in [0.50, 0.94]$.
    Panel~(a) shows the total throughput summed over all $n$ paths;
    panel~(b) shows the throughput normalised per SD pair.
    QLO serves as an infeasible upper bound since it ignores the shared link conflict resolution. MIN, GRO, and NUM all enforce feasibility; the gap between them and
    REF reflects the cost of satisfying $F_{ete}^{th}$.}
\label{fig:observed_throughput_multi_range_9094_5094}
\end{figure*}

Three end-to-end fidelity threshold ranges are evaluated for this study:
(i)~$F_{ete}^{th} \in [0.90, 0.94]$;
(ii)~$F_{ete}^{th} \in [0.70, 0.94]$; and
(iii)~$F_{ete}^{th} \in [0.50, 0.94]$.
\subsection{Approaches evaluated}\label{ssec:approaches}
Five approaches are compared, summarised in Table~\ref{tab:comparison}.

\begin{enumerate}
\item \textbf{REF} (Reference baseline): each link is operated at its
individual threshold $a_e = a_e^{th}$, ignoring the end-to-end
fidelity constraint entirely. REF maximises raw activation time and serves as an infeasible upper bound on throughput.

\item \textbf{QLO} (per-path, Theorem~\ref{theorem.chain}): Algorithm~\ref{alg:QLO}
is applied independently to each path, accounting for both the end-to-end
fidelity constraint and per-link thresholds. However, in a general network,
different paths may assign conflicting activation periods to the same
physical shared link. QLO does not resolve these conflicts; each path uses
its own independently optimal values on shared links. The reported throughput
is therefore physically infeasible in any network with shared links,
constituting a second infeasible upper bound that is tighter than REF but
unachievable in practice.

\item \textbf{MIN} (Minimum selection rule, Algorithm~\ref{alg:multi-path}):
per-path QLO allocations are computed first; shared links are then clamped
to the minimum across all paths that traverse them.
MIN is feasible by construction
(Proposition~\ref{prop:general-overlap}) and serves as the conservative lower
bound among the constrained approaches.

\item \textbf{GRO} (Greedy quantum link orchestration,
Algorithm~\ref{alg:gro_hybrid}): extends MIN by redistributing the freed
calibration budget on each affected path over its exclusive (unshared) links.
GRO is feasible and weakly dominates MIN
(Proposition~\ref{prop:gro-feasibility}).

\item \textbf{NUM} (Numerical optimisation): the sum of path throughputs
$T^{NUM} = \sum_\pi \mathcal{T}_\pi$ is maximised directly, subject to the linear
fidelity-budget constraints, using the sequential quadratic programming solver
(SLSQP) from SciPy, warm-started from the GRO solution. A single physical
activation period is enforced per shared link across all paths, making NUM a
true constrained optimiser of the general network problem. The warm start
guarantees $T^{\mathrm{NUM}} \ge T^{\mathrm{GRO}}$ by construction. Being
feasible, NUM lower-bounds the true network optimum $T^\star$, while QLO
upper-bounds it, so the two bracket it as
$T^{\mathrm{GRO}} \le T^{\mathrm{NUM}} \le T^\star \le T^{\mathrm{QLO}}$; SLSQP
attains only a local optimum of this non-convex problem, hence NUM is the
strongest feasible reference against which GRO is evaluated. This accuracy is
costly: the warm-started SLSQP solves make the full NUM sweep take hours,
against minutes for GRO on the identical harness (Section~\ref{ssec:results}).
NUM is therefore an offline benchmark rather than a deployable strategy; GRO is
the practical alternative.
\end{enumerate}

\begin{figure*}[!t]
\centering
\includegraphics[width=2\columnwidth]{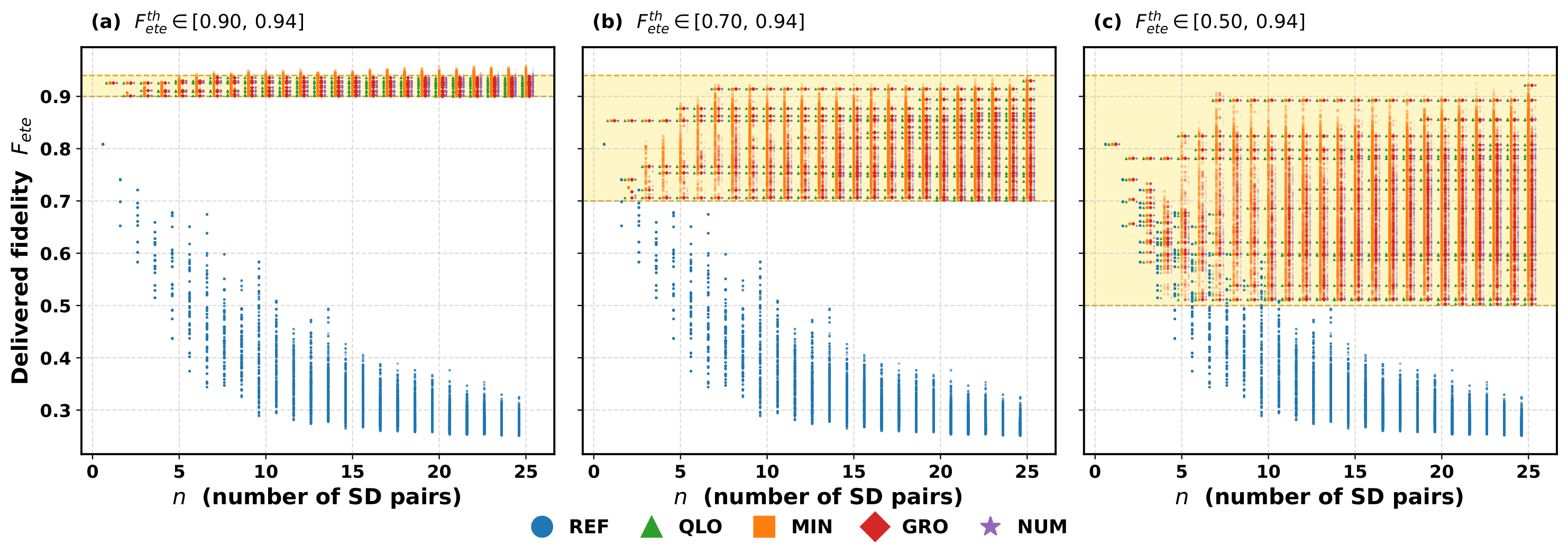}
\caption{Individual delivered end-to-end fidelity for each number of SD pairs $n$ for all approaches. The depicted yellow range is the end-to-end fidelity thresholds set with panel (a) $F_{ete}^{th} \in [0.90, 0.94]$; panel (b) $F_{ete}^{th} \in [0.70, 0.94]$; panel (c) $F_{ete}^{th} \in [0.50, 0.94]$. REF shows a similar pattern across different ranges since it does not account for end-to-end fidelity thresholds, while other approaches all deliver, considering the end-to-end fidelity thresholds.}
\label{fig:fidelity_scatter_combine}
\end{figure*}

\subsection{Simulation results}\label{ssec:results}
By construction of the five approaches, the sum of the path throughputs ($T$) exhibits the ordering
\begin{equation}
T^{\mathrm{REF}} \;\ge\; T^{\mathrm{QLO}} \;\ge\; T^{\mathrm{NUM}}
\;\ge\; T^{\mathrm{GRO}} \;\ge\; T^{\mathrm{MIN}}.
\label{eq:ordering}
\end{equation}
This ordering holds for every $K>0$. Since QLO, MIN, and GRO allocate in $L$-space
independently of $K$; each path's QLO allocation is optimal for that path
alone, so no consistent allocation can exceed it; GRO only increases $L$
relative to MIN; and NUM is warm-started from GRO. $K$ therefore changes the
size of the gaps in Table~\ref{tab:rel_throughput_vs_qlo}, not their sign.
Fig.~\ref{fig:observed_throughput_multi_range_9094_5094} shows the
throughput as a function of the number of SD pairs $n$ for all five approaches
across the fidelity ranges ($[0.90, 0.94]$), ($[0.70, 0.94]$), and ($[0.50, 0.94]$).
\begin{remark}
In Fig.~\ref{fig:observed_throughput_multi_range_9094_5094}, REF exhibits the lowest performance, which appears to contradict the expected ordering in Eq.~\eqref{eq:ordering}. This discrepancy arises because the plot excludes the throughput of paths whose end-to-end fidelity is below the required threshold, thereby reflecting realistic operational constraints. This behaviour is further illustrated in Fig.~\ref{fig:fidelity_scatter_combine}, where a significant fraction of the delivered fidelity does not meet the end-to-end fidelity threshold (Fig.~\ref{fig:fidelity_scatter_combine}a). The limited throughput observed for REF is therefore primarily due to the relaxation of the end-to-end fidelity threshold (Fig.~\ref{fig:fidelity_scatter_combine}b and \ref{fig:fidelity_scatter_combine}c), consistent with the trends shown in Fig.~\ref{fig:observed_throughput_multi_range_9094_5094}.
\end{remark}
This ordering is structurally meaningful: REF is the infeasible ceiling since
it ignores the end-to-end constraint; QLO is a second infeasible upper bound
since it ignores shared-link conflicts; and NUM, GRO, MIN are the three
feasible approaches ordered by the degree of budget recovery they achieve.

The gap between QLO and NUM represents the \emph{cost of shared-link
coordination}: the throughput that must be surrendered when each physical link
is constrained to a single activation period shared across all paths that
traverse it. At small $n$, most links
are exclusive to individual paths, and the coordination cost is negligible;
at $n = 25$, the gap between QLO and NUM becomes substantial, reflecting a significant increase in the links shared by multiple paths.

Fig.~\ref{fig:observed_throughput_multi_range_9094_5094}b shows the
per-SD-pair throughput, which captures a complementary effect: as $n$ grows
and shared links accumulate more contention, the average throughput available
per path decreases. The per-path throughput of GRO and MIN degrades more
steeply than that of NUM, confirming that the joint optimisation performed
by NUM provides increasing benefit at scale.

\begin{table*}[!t]
\centering
\caption{MIN, GRO, and NUM throughput relative to QLO (\%) for three
end-to-end fidelity threshold ranges and representative values of~$n$.
Negative values indicate degradation with respect to the per-path QLO
upper bound.}
\label{tab:rel_throughput_vs_qlo}
\renewcommand{\arraystretch}{1.15}
\setlength{\tabcolsep}{5pt}
\begin{tabular}{c|rrr|rrr|rrr}
\hline
& \multicolumn{3}{c|}{$F_{ete}^{th} \in [0.90,\,0.94]$}
& \multicolumn{3}{c|}{$F_{ete}^{th} \in [0.70,\,0.94]$}
& \multicolumn{3}{c}{$F_{ete}^{th} \in [0.50,\,0.94]$} \\
$n$
& \multicolumn{1}{c}{MIN}
& \multicolumn{1}{c}{GRO}
& \multicolumn{1}{c|}{NUM}
& \multicolumn{1}{c}{MIN}
& \multicolumn{1}{c}{GRO}
& \multicolumn{1}{c|}{NUM}
& \multicolumn{1}{c}{MIN}
& \multicolumn{1}{c}{GRO}
& \multicolumn{1}{c}{NUM} \\
\hline
 1  &    0.00 &    0.00 &    0.00 &    0.00 &    0.00 &    0.00 &    0.00 &    0.00 &    0.00 \\
 5  &  $-$0.33 &  $-$0.21 &  $-$0.11 &  $-$0.26 &  $-$0.21 &  $-$0.15 &  $-$0.17 &  $-$0.15 &  $-$0.12 \\
10  &  $-$2.64 &  $-$2.11 &  $-$1.29 &  $-$2.38 &  $-$2.15 &  $-$1.65 &  $-$1.79 &  $-$1.67 &  $-$1.36 \\
15  &  $-$8.29 &  $-$7.14 &  $-$4.47 &  $-$6.13 &  $-$5.69 &  $-$4.21 &  $-$4.57 &  $-$4.33 &  $-$3.37 \\
20  & $-$19.18 & $-$17.51 & $-$10.47 & $-$13.42 & $-$12.74 &  $-$8.73 &  $-$9.92 &  $-$9.52 &  $-$6.96 \\
25  & $-$35.37 & $-$33.56 & $-$17.48 & $-$28.63 & $-$27.83 & $-$16.79 & $-$23.18 & $-$22.66 & $-$14.56 \\
\hline
\end{tabular}
\end{table*}

Fig.~\ref{fig:fidelity_scatter_combine} shows the delivered end-to-end
fidelity of every individual path across all combinations and all values of
$n$, for three fidelity threshold ranges. The gold-shaded band marks the
range $[F_{th,\text{low}}, F_{th,\text{high}}]$ from which the per-path
thresholds were drawn.

The REF approach delivers fidelity values that are largely distributed
independently of the threshold band and show a similar scatter pattern across
all three ranges. This is expected: REF operates each link at its
threshold $a_e^{th}$ without consulting the end-to-end fidelity requirement,
so the delivered fidelity is determined purely by the network topology and
the per-link parameters rather than by any path-level constraint.

In contrast, QLO, MIN, GRO, and NUM all enforce the end-to-end fidelity
constraint, and their scatter distributions are visibly concentrated at or
above the lower edge of the threshold band. The fraction of points falling
below the band (fidelity violations) is none for QLO, MIN, GRO, and NUM by
construction.

Comparing across ranges, the widening of the threshold band from
$[0.90, 0.94]$ (panel~a) to $[0.50, 0.94]$ (panel~c) produces a visibly
larger spread in the delivered fidelity for the constrained approaches,
reflecting the broader range of per-path requirements. The feasible approaches
deliver fidelity values distributed across the full threshold band in each
case, confirming that the calibration optimisation is active across the
entire range of requirements.
Because REF, QLO, MIN and GRO allocate in $L$-space, their delivered
fidelities are independent of $\Gamma$ and $c$; only NUM, which maximises
throughput directly, depends on $K$.

\subsection{MIN, GRO, and NUM performance relative to QLO}\label{ssec:relative_performance}

As seen in Sec.~\ref{ssec:results}, each of the approaches results in a unique allocation, including QLO that allocates optimally for each path treated independently. We therefore treat QLO as a reference metric against which the other approaches are evaluated.  Table~\ref{tab:rel_throughput_vs_qlo} reports the throughput of MIN, GRO,
and NUM as a percentage deviation from the per-path QLO upper bound as
\begin{equation}
    \Delta\% = \frac{T_{approach} - T_{QLO}}{T_{QLO}} \times 100,
\end{equation}
$\forall \ \text{approach}\in \{\text{MIN, GRO, NUM}\}$, across
all three fidelity threshold ranges and representative values of~$n$.

At $n=1$, all approaches coincide with QLO, with a single SD pair there
are no shared links, so the minimum selection rule, GRO, and the numerical
optimiser all reduce to the single-path QLO solution and the deviation
is~$0\%$. As $n$ increases, all three feasible approaches degrade below QLO,
but at markedly different rates, consistently ordered as
$T^{\mathrm{NUM}} \ge T^{\mathrm{GRO}} \ge T^{\mathrm{MIN}}$,
which is the expected ordering from Eq.~\eqref{eq:ordering}.

\section{Conclusion}\label{sec:conclusion}

This work addressed the problem of calibration-aware entanglement distribution in
quantum networks, motivated by experimental evidence that the initial fidelity of
entangled pairs generated over fiber links decay monotonically with the duration
of the activation period and is restored only through periodic recalibration.
We modelled this duty-cycle structure explicitly at the network layer and
formulated the resulting throughput optimisation as a constrained allocation
problem over activation periods.

For linear quantum chains, we derived the optimal allocation analytically.
Theorem~\ref{theorem.chain} (Quantum Link Orchestration, QLO) shows that the
throughput-maximising allocation follows a closed-form partitioning rule: links
whose per-link fidelity threshold falls below the path-optimal point $\Omega_\pi$
are operated at their individual threshold; links above $\Omega_\pi$ are clamped
to a recursively updated optimal point that redistributes the remaining
end-to-end budget. The algorithm runs in $O(N^2)$ time and the optimal point is
non-decreasing across recursion levels, guaranteeing finite termination and
optimality.

Extending to general network topologies, where multiple paths may traverse the
same physical link introduces a conflict on the selection of activation periods. We established
the minimum selection rule (MIN, Algorithm~\ref{alg:multi-path}) as a feasible
$O(M)$ benchmark: shared links are assigned the minimum QLO allocation across all
paths that traverse them, guaranteeing end-to-end fidelity compliance for every
path by construction. Building on this, we proposed the Greedy Quantum Link
Orchestration heuristic (GRO, Algorithm~\ref{alg:gro_hybrid}), which recovers the
budget freed by shared-link clamping and redistributes it over each path's exclusive
links, sorted by descending available headroom. GRO runs in
$O(P \cdot N_\pi^2 + M \log M)$ and weakly dominates MIN with strict improvement
whenever freed budget and positive headroom coexist.


Several directions remain open. First, the GRO heuristic processes paths
independently in Pass~2; a coordinated multi-path reallocation that accounts for
the joint throughput objective could narrow the GRO--NUM gap at large $n$ without
the full cost of numerical optimisation. Second, this work assumed that the
shortest path between each SD pair is fixed, and the link parameters are
homogeneous; relaxing these to allow path selection and heterogeneous decay rates
$\Gamma_e$ would bring the model closer to deployed quantum network conditions.
Third, the fidelity decay model validated here is a single-exponential model; incorporating the stretched-exponential fit observed in the
experimental data of Fig.~\ref{fig:fid_vs_time} would improve accuracy for long
activation periods.

\section*{Acknowledgment}

The authors would like to thank Victor Krutyanskiy for providing experimental data of the optical-fiber link at the University of Innsbruck.

The author(s) acknowledge the use of Claude (Sonnet 4.6, Opus 4.6) for code generation and code presentation for a few parts of the codebase. 

All content was reviewed and edited by the authors, who take full responsibility for the final work.

\appendices

\section{Processing of Experimental Data and Derivation of Fig.~\ref{fig:fid_vs_time}}
\label{appendix:experimental_data}

\subsection*{A. Raw data format}

The experimental data were kindly provided by Dr. Victor Krutyanskiy from the
Innsbruck group. The measurements were performed in 2020 on the same 520\,m
fiber link that was later used for the ion-ion entanglement experiment
of~\cite{viktor2023trappedionentanglement}. The experiment ran continuously
from 17:59 on 3 June 2020 to 14:40 on 4 June 2020, spanning approximately
$74{,}400$\,s ($\sim$20.7 hours) in total. Each record in the dataset takes
the form
\[
\begin{split}
(\texttt{Proj},\ t_{\mathrm{stamp}},\ t,\ 
\texttt{Ax1}^{\mathrm{abs}},\ \texttt{Ax1}^{\mathrm{rel}},\\
\texttt{Ax2}^{\mathrm{abs}},\ \texttt{Ax2}^{\mathrm{rel}},\\
S_1,\ S_2,\ S_3,\ \eta,\ \theta,\ \mathrm{DOP},\ I)
\end{split}
\]
where the column \texttt{Proj} $\in \{H, L, A, R, V, D\}$
labels the six input polarization states iteratively prepared and injected into the fiber:
horizontal linear (H), left circular (L), antidiagonal (A), right circular (R),
vertical linear (V), and diagonal (D), cycled in this order. For each of these input states at the recorded wall-clock timestamp ($t_{\mathrm{stamp}}$) and the elapsed time from the start of the run in seconds ($t$) the polarisation state coming out of the other end of the fiber is determined. $S_1$, $S_2$, $S_3$ are the three non-trivial
Stokes parameters of the output polarization state. For the fidelity
analysis only the columns $(\texttt{Proj},\ t,\ S_1, S_2, S_3)$ are used; the
remaining columns provide auxiliary relavant to the experiment.

\subsection*{B. Stokes parameters and the Poincar\'{e} sphere}

A fully polarized optical field is completely described by its Stokes vector
$(S_0, S_1, S_2, S_3)$. The Stokes parameters are the direct
optical analogue of the three components of the Bloch vector for a qubit:
$S_1$ corresponds to the $Z$-axis (H/V basis), $S_2$ to the $X$-axis (D/A
basis), and $S_3$ to the $Y$-axis (R/L basis) \cite{huistokes2009}. Considering transmitted states
are close to pure we set $S_0 = 1$
throughout, so that the state lives on the unit Poincar\'{e} sphere with
\[
    S_1^2 + S_2^2 + S_3^2 = p^2 \;\le\; 1,
\]
where $p = \sqrt{S_1^2 + S_2^2 + S_3^2}$ is the degree of polarization.
The reduced Stokes parameters \cite{huistokes2009}
\begin{equation}
    Q = \frac{S_1}{S_0\,p}, \qquad
    U = \frac{S_2}{S_0\,p}, \qquad
    V = \frac{S_3}{S_0\,p}
    \label{eq:QUV}
\end{equation}
parametrise the surface of the unit sphere.

\subsection*{C. Stokes-to-Jones conversion}

For a pure state the Jones vector
$\mathbf{j} = (j_1,\, j_2)^{\mathsf{T}} \in \mathbb{C}^2$ is related to the
Stokes parameters through \cite{huistokes2009}
\begin{align}
    j_1 &= \sqrt{\frac{1+Q}{2}}, \label{eq:jones1} \\[4pt]
    j_2 &= \frac{U}{2\,j_1} - i\,\frac{V}{2\,j_1}, \label{eq:jones2}
\end{align}
where $Q$, $U$, $V$ are as defined in Eq.~\eqref{eq:QUV}. When $j_1 = 0$, the
expression in Eq.~\eqref{eq:jones2} is undefined; in this limiting case
one sets $j_2 = 1$ and $j_1 = 0$ directly from the spherical coordinates.

\subsection*{D. Necessity of six projections}

Measuring the drift of a single input polarization state over time is
insufficient to characterize the full polarization transformation induced by
the fiber. Slow birefringence drift corresponds to a time-varying unitary
rotation on the Poincar\'{e} sphere. Certain rotations leave specific polarization
states invariant, for instance, a rotation about the $S_1$ axis does not
change the H state but substantially rotates the D and R states \cite{huistokes2009}. More
generally, any polarization state that is an eigenstate of the instantaneous
fiber rotation operator will appear unchanged at the output even as the fiber
drifts. Measuring all six Stokes basis projections $\{H, V, D, A, R, L\}$
provides a redundant and robust detection of any unitary drift, since no nontrivial unitary rotation can simultaneously fix all six states. The full dataset is sufficient (when full Stokes vectors are measured) to reconstruct the $4\times4$ Mueller matrix of the fiber transformation at each time $t$. However for this work, we
do not require the full process tomography; instead, we compute a per-projection
fidelity and average across projections to obtain a scalar proxy for the overall
polarization stability of the link.

\subsection*{E. Per-projection fidelity}

For each projection $k \in \{H, L, A, R, V, D\}$, a reference Jones vector
$\mathbf{j}^{(k)}_{\mathrm{ref}}$ is established from the \emph{first}
measurement of that projection in the dataset, corresponding to $t \approx 0$
when the link is freshly calibrated and the polarization alignment is at its
best. The instantaneous polarization fidelity at time $t$ for projection $k$
is then
\begin{equation}
    F^{(k)}(t) \;=\; \left|\bigl\langle \mathbf{j}^{(k)}_{\mathrm{ref}},\;
    \mathbf{j}^{(k)}(t) \bigr\rangle\right|^2,
    \label{eq:fidelity_jones}
\end{equation}
where $\langle \cdot, \cdot \rangle$ denotes the standard Hermitian inner
product on $\mathbb{C}^2$. Because $\mathbf{j}^{(k)}_{\mathrm{ref}}$ is
normalized, $F^{(k)}(t) \in [0,1]$, with $F^{(k)}(0) = 1$ by construction.
This quantity measures how much the transmitted polarization state has drifted
from its initial alignment, and therefore directly proxies the fidelity of the
distributed entangled pair generated by that link at time $t$.

\subsection*{F. Time alignment and cross-projection averaging}

The six projections are recorded in a fixed round-robin cycle
$H \to L \to A \to R \to V \to D$. This introduces slightly staggered
absolute start times across projections. To place all traces on a common time
axis, the H projection is taken as the reference, and every other projection
$k$ has its time axis shifted by
\[
    \Delta t^{(k)} = t^{(H)}_0 - t^{(k)}_0,
\]
where $t^{(k)}_0$ is the first recorded timestamp for projection $k$. After
alignment, let $\mathcal{T}^{(k)} = \{t^{(k)}_1, t^{(k)}_2, \ldots\}$ denote
the ordered sequence of measurement times for projection $k$. Each trace is
sorted chronologically and all six traces are trimmed to the same number of
samples $N_{\min} = \min_k |\mathcal{T}^{(k)}|$. The
averaged fidelity at each sample index $m$ is
\begin{equation}
    \bar{t}_m = \frac{1}{6}\sum_{k} t^{(k)}_m, \qquad
    \bar{F}_m = \frac{1}{6}\sum_{k} F^{(k)}_m,
    \label{eq:avg_fidelity}
\end{equation}
yielding the single averaged curve $\bar{F}(\bar{t})$ plotted in
Fig.~\ref{fig:fid_vs_time}b.

\subsection*{G. Model selection via Akaike Information Criterion (AIC)}

To characterise the fidelity decay analytically, three parametric families are
fitted to $\bar{F}(\bar{t})$ over the window $[0,\, 20{,}000]$\,s using
robust nonlinear least squares (soft-$\ell_1$) loss\footnote{
The soft-$\ell_1$ loss replaces the standard squared-residual penalty with
$\rho(r) = 2\!\left(\sqrt{1+r^2}-1\right)$,
which behaves as $\rho(r)\approx r^2$ for small residuals and as
$\rho(r)\approx 2|r|$ for large ones. This smooth transition from quadratic
to linear penalization reduces the influence of occasional outlying
polarimetry readings on the fitted decay parameters without discarding them
entirely, unlike hard truncation. In practice, residuals larger than the
scale parameter $f_{\mathrm{scale}}=0.02$ (in fidelity units) are treated
as outliers and down-weighted accordingly. The fitting is implemented via
\texttt{scipy.optimize.least\_squares} with \texttt{loss=`soft\_l1'} and
\texttt{f\_scale=0.02}.}. The fitted data covers
roughly the first quarter of the full $74{,}400$\,s dataset. The selected window is chosen because
the fidelity reaches its asymptotic floor by the last point of the window. The three parametric families are:
\begin{enumerate}
    \item \textbf{Single exponential} ($k=3$ parameters):
    \[
        F(t) = c + a\,e^{-t/\tau};
    \]
    \item \textbf{Bi-exponential} ($k=5$ parameters):
    \[
        F(t) = c + a_1\,e^{-t/\tau_1} + a_2\,e^{-t/\tau_2};
    \]
    \item \textbf{Stretched exponential} ($k=4$ parameters):
    \[
        F(t) = c + a\,e^{-(t/\tau)^{\beta}}.
    \]
\end{enumerate}
Model selection is performed using the Akaike Information
Criterion (AIC)~\cite{akaike1974aic}. For nonlinear least squares with
independent Gaussian residuals of unknown variance, maximising the
likelihood over $\sigma^2$ and discarding the additive constants common to
all candidate models yields the standard working form~\cite{burnham2002aic}
\begin{equation}
    \mathrm{AIC} = n\ln\!\left(\frac{\mathrm{RSS}}{n}\right) + 2k,
    \label{eq:aic}
\end{equation}
where $\mathrm{RSS}=\sum_m r_m^2$ is the residual sum of squares, $n$ the
number of samples and $k$ the number of free parameters. The model with the
lowest AIC is preferred, as it balances goodness of fit against complexity.

Across the experimental dataset, the stretched exponential consistently
achieves the lowest AIC, yielding the best-fit parameters
\[
    c = 0.333, \quad a = 0.672, \quad
    \tau = 10{,}527\,\text{s}, \quad \beta = 3.00.
\]
The value $\beta = 3$ produces a characteristically flat shoulder at early
times followed by a sharp decay: the fidelity remains above $0.9$ for
approximately the first $7{,}000$\,s before falling steeply. The
single-exponential model used throughout the main body of this work
(Assumption~\ref{assumption.initialfidelityexponential}) is therefore a
conservative approximation: it predicts faster decay than is observed at
intermediate times, yielding a shorter estimated permissible activation
period and hence a conservative bound on throughput.

\section{Proof of Theorem~\ref{theorem.equalactivation}}\label{appendix.proofs}
\begin{proof}
    From Assumption~\ref{assumption.endtoendfidelity} with $\mathcal{U} = 1$ we have
\begin{equation}\label{eq:appendixb_initial}
    F_{ete} = \frac{1}{4} \left[ 1 + 3 \prod_{i=0}^{N-1} \left( \frac{4 \ F_i - 1}{3} \right) \right].
\end{equation}
Let the quantum network require an end-to-end fidelity threshold of $F_{ete}^{th}.$ So
\begin{equation}
    F_{ete}^{th} \leq F_{ete} = \frac{1}{4} \left[ 1 + 3\prod_{i=0}^{N-1} \left( \frac{4 \ F_i - 1}{3} \right) \right]
\end{equation}
\begin{equation}
    F_{ete}^{th} \leq \frac{1}{4}\left[ 1 + 3\prod_{i=0}^{N-1} \left( \frac{4 \ F_i - 1}{3} \right) \right]
\end{equation}
\begin{equation}
    \Rightarrow 3^N \left( \frac{4 F_{ete}^{th} -1}{3} \right) \leq \prod_{i=0}^{N-1} \left(4 F_i -1\right)
\end{equation}
\begin{equation}
    \Rightarrow \frac{4^N}{3^N} \left( \frac{3}{4 F_{ete}^{th} -1} \right)  \geq \frac{ 1}{ \prod_{i=0}^{N-1} \left( F_i - \frac{1}{4} \right)}.
\end{equation}
Multiplying by $\left( F_e^M - \frac{1}{4} \right)^N$ on both sides implies
\begin{equation}
     \frac{4^N}{3^N} \left( \frac{3}{4 F_{ete}^{th} -1} \right) \left( F_e^M - \frac{1}{4} \right)^N \geq \frac{ \left( F_e^M - \frac{1}{4} \right)^N}{ \prod_{i=0}^{N-1} \left( F_i - \frac{1}{4} \right)}.
\end{equation}
Taking logarithms on both sides and using the definition of $L_i = \ln{\left( \frac{ F_e^M - \frac{1}{4}}{F_i - \frac{1}{4}} \right)}$ and $A =  F_e^M - \frac{1}{4}$, we have
\begin{equation}
    \ln{\left[ \frac{4^{N-1}}{3^{N-1}} \frac{\left( F_e^M - \frac{1}{4}\right)^N}{F_{ete}^{th} - \frac{1}{4}} \right] } \geq \sum_{i=0}^{N-1} L_i,
\end{equation}
\begin{equation}\label{eq:sumLi}
    \boxed{\sum_{i=0}^{N-1} L_i \leq (N-1) \ln{\left( \frac{4}{3} \right)} + N \ln{A} - \ln{ \left( F_{ete}^{th} - \frac{1}{4} \right)}    }.
\end{equation}
In the above equation, the LHS denotes the contribution from all the links, with an upper bound given by the RHS, which is the budget $\mathcal{L}$ in Problem~\ref{prob:problem1}. With imperfect swapping operations, the same steps applied to Assumption~\ref{assumption.endtoendfidelity} give the additional term $\ln\mathcal{U}$ on the RHS, and with heterogeneous maximum fidelities $N\ln A$ becomes $\sum_i\ln A_i$. Here, the derivation only involves the initial fidelities $F_i$, which is independent of the drift law that maps activation periods to fidelities. Now consider the case of no per-link fidelity threshold $L_i^{th}$ for all $i$. Since $\phi(L)=\ln\frac{L}{L+K}$ is strictly increasing and strictly concave (Remark~\ref{remark.concave}), Jensen's inequality gives, for every allocation satisfying Eq.~\eqref{eq:sumLi},
\begin{equation}
\sum_{i=0}^{N-1}\phi(L_i)\;\le\;N\,\phi\Bigl(\frac1N\sum_{i=0}^{N-1}L_i\Bigr)\;\le\;N\,\phi\Bigl(\frac{\mathcal{L}}{N}\Bigr),
\end{equation}
where the first inequality is an equality if and only if all $L_i$ are equal, and the second if and only if $\sum_i L_i=\mathcal{L}$. Hence, the network throughput $\ln\mathcal{T}=\ln\mathcal{T}_{nc}+\sum_i\phi(L_i)$ is maximised if and only if every link is allocated the equal share $L_e=\mathcal{L}/N$, which is feasible because of no individual link fidelity threshold and hence,
\begin{equation}\label{eq:omegaL}
    \boxed{\Omega_L = L_e = \left( \frac{N-1}{N} \right) \ln{ \left( \frac{4}{3} \right)} + \ln{(A)}-\frac{1}{N} \ln{\left( F_{ete}^{th} - \frac{1}{4}\right)}}.
\end{equation}
\end{proof}

\section{Optimal point at $\gamma+1$ layer of optimization is greater than or equal to at $\gamma$ layer}\label{appendix:optimal_point_proof}
\begin{proof} 
From Theorem~\ref{theorem.chain}, the recurrence relation between the optimal point is
\begin{equation}
    \Omega_{\pi}^{\gamma+1} = \frac{1}{k} \left[ 
 (j+k) \; \Omega_{\pi}^{\gamma} - \sum_j L_j^{th} \right],
\end{equation}
\begin{equation}\label{eq.recurrence_optimal_reduced}
    \Rightarrow k \; \Omega_{\pi}^{\gamma+1}=(j+k) \;\Omega_{\pi}^{\gamma} - \sum_j L_j^{th}.
\end{equation}
Now since $L_j^{th} \leq \Omega_{\pi}^{\gamma}$, let $L_j^{th} = \Omega_{\pi}^{\gamma} - \epsilon_j$ where $\epsilon_j \geq 0$ dictates individual closeness of initial fidelity thresholds of each link to the optimal point of the path $\pi$.

From Eq.~\eqref{eq.recurrence_optimal_reduced} we get,
\begin{equation}
    k \; \Omega_{\pi}^{\gamma+1} = (j+k) \; \Omega_{\pi}^{\gamma} - \sum_j \left( \Omega_{\pi}^{\gamma} - \epsilon_j \right)
\end{equation}
\begin{equation}
    \Rightarrow k \; \Omega_{\pi}^{\gamma+1} = (j+k) \; \Omega_{\pi}^{\gamma} - j \; \Omega_{\pi}^{\gamma} + \sum_j \epsilon_j
\end{equation}
\begin{equation}
    \Rightarrow \Omega_{\pi}^{\gamma+1} = \Omega_{\pi}^{\gamma} + \frac{1}{k} \left( \sum_j \epsilon_j \right)
\end{equation}
\begin{equation}
\boxed{
    \Omega_{\pi}^{\gamma+1} \geq \Omega_{\pi}^{\gamma} }.
\end{equation}
\end{proof}

\section{Two-link quantum linear chain}\label{sec:two_link_example}
Consider a two-link linear quantum chain, and visualise the
calibration-aware throughput optimisation problem under both the initial
fidelity and the end-to-end fidelity constraints. Utilising
Eq.~\eqref{eq:end_to_end_fidelity} for a two-link chain, the end-to-end
fidelity of EPR pair delivery in terms of the initial fidelity of EPR pair
generation at each link is
\begin{equation}
    F_{ete} = \frac{1}{4} \left[ 1 + 3 \left( \frac{4 F_0^m - 1}{3} \right)
    \left( \frac{4 F_1^m - 1}{3} \right) \right].
\end{equation}
For an end-to-end threshold of $0.8$, the admissible pairs $(F_0^m, F_1^m)$
form the region bounded by
$F_1^m = \tfrac{1}{4}\bigl(1 + \tfrac{33}{5(4F_0^m-1)}\bigr)$. Its symmetric
point $\Omega_F = (0.8923,\,0.8923)$ is the throughput maximum and is
precisely where $a_{e0} = a_{e1} = a_{ete}$, by
Lemmas~\ref{lemma.initialfidelityvariation}
and~\ref{lemma.endtoendfidelityactivation}; moving along the boundary away
from it breaks that equality in one direction or the other. Raising the
end-to-end threshold shrinks the admissible region.
Fig.~\ref{fig:heatT_combined} gives the throughput heat map over this region
in (a) the linear fidelity space and (b) the corresponding L-space.
\begin{figure*}[!t]
\centering
\includegraphics[width=2\columnwidth]{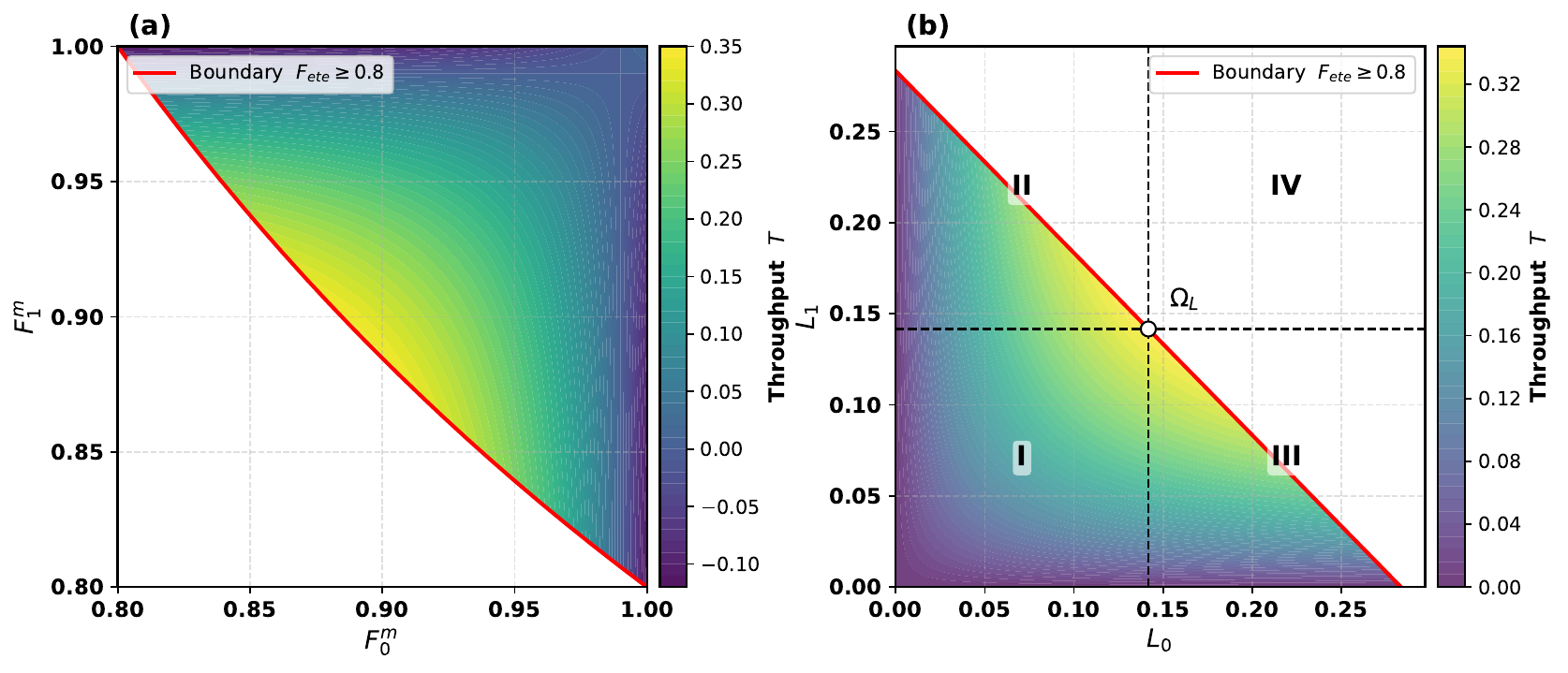}
\caption{Throughput heat maps of allowed regions for a two-link linear quantum network chain under end-to-end fidelity constraint visualised in (a) linear fidelity space, (b) Log-fidelity space.}
\label{fig:heatT_combined}
\end{figure*}
A more informative way of visualising allowed regions for a two-link network chain is done in Fig.~\ref{fig:heatT_combined}, which gives the heat map of the throughput achieved for each operating point in the allowed region. The central point on the boundary is the point of maximum throughput $\Omega$ or $\Omega_L$(marked in Fig.~\ref{fig:heatT_combined}b). Given that we now have the heat map for the throughput visually present, we can write the constraint of our throughput optimisation problem and proceed visually towards tackling the problem. 
Using Theorem~\ref{theorem.equalactivation}, for a two-link quantum linear chain, the allowed region in L-space is given by (see Eq.~\eqref{eq:sumLi} in Appendix~\ref{appendix.proofs}):
\begin{equation}
   \text{C1:} \ L_0 + L_1 \leq \ln{ \left( A^2 \frac{80}{33} \right) }.
\end{equation}
Remember that this equation is the manifestation of the end-to-end fidelity constraint in L-space.
Also, let the constraint on the initial fidelity of the generated EPR pair for each link as:
\begin{equation}
\begin{aligned}
    \text{C2:} \ L_0 &\leq L_0^{th}, \\
    \text{C3:} \ L_1 &\leq L_1^{th}.
\end{aligned}
\end{equation}
Now let us consider cases corresponding to each quadrant of the Fig.~\ref{fig:heatT_combined}b.
\begin{itemize}
    \item \textbf{Case $I$:} For $L_0^{th} \leq \Omega_L$ and $L_1^{th} \leq \Omega_L$, Constraint C1 is always satisfied, so according to Theorem~\ref{theorem.optimalthroughput}, the throughput is maximised while operating at the thresholds, i.e., choose $L_0 = L_0^{th}$ and $L_1 = L_1^{th}$.
    \item \textbf{Case $II$:} For $L_0^{th} < \Omega_L$ and $\Omega_L < L_1^{th}$, as seen in the Fig.~\ref{fig:heatT_combined}b, the area under this quadrant consists of allowed as well as not-allowed region (i.e., satisfy or violate Constraint C1). Hence, we have the following two sub-cases:
    \begin{enumerate}
        \item \textbf{Sub-case I:} If it lies in the allowed region, then the Constraint C1 is satisfied by definition. Then similar to case $I$ choose $L_0 = L_0^{th}$ and $L_1 = L_1^{th}.$
        \item \textbf{Sub-case II:} If it does not lie in the allowed region, Constraint C1 is violated. Then to maximize the throughput choose $L_0 = L_0^{th}$ while $L_1$ on the boundary of the allowed region at $L_0^{th}.$
    \end{enumerate}
    \item \textbf{Case $III$:} For $\Omega_L < L_0^{th}$ and $L_1^{th} < \Omega_L$, variable reverse case of case $II$.

    \item \textbf{Case $IV:$} For $\Omega_L < L_0^{th}$ and $\Omega_L < L_1^{th}$, this quadrant is entirely outside the allowed region, and it violates Constraint C1. To maximise the throughput, we simply choose $L_0 = L_1 = \Omega_L$, which is the point of highest throughput.
\end{itemize}
The above allocation equates to the general optimal allocation schedule sketched in Theorem~\ref{theorem.chain}.

\section{Minimum selection rule benchmark}\label{appendix.minimum_rule}
\begin{algorithm}[!t]
\caption{Minimum selection rule benchmark for shared-link calibration}
\label{alg:multi-path}
\begin{algorithmic}[1]
\REQUIRE Set of paths $P$; per-path QLO allocations
         $L_{\{\pi,e\}}$ for each $\pi \in P$, $e \in \pi$
\ENSURE Per-path link assignment $\{L_e^\pi : \pi \in P,\, e \in \pi\}$

\STATE \textit{// Pass 1: compute the running minimum per link across all paths}
\STATE Initialise $\mathrm{min\_val}[e] \gets +\infty$ for all $e$
\FORALL{$\pi \in P$}
  \FORALL{$e \in \pi$}
    \STATE $\mathrm{min\_val}[e] \gets \min\!\bigl(\mathrm{min\_val}[e],\; L_{\{\pi,e\}}\bigr)$
  \ENDFOR
\ENDFOR

\STATE \textit{// Pass 2: assign shared links their minimum; exclusive links are unchanged}
\FORALL{$\pi \in P$}
  \FORALL{$e \in \pi$}
    \IF{$|\Pi_e| \ge 2$}
      \STATE $L_e^\pi \gets \mathrm{min\_val}[e]$
    \ELSE
      \STATE $L_e^\pi \gets L_{\{\pi,e\}}$
    \ENDIF
  \ENDFOR
\ENDFOR

\RETURN $\{L_e^\pi : \pi \in P,\, e \in \pi\}$
\end{algorithmic}
\end{algorithm}

\begin{corollary}[Time complexity of Algorithm~\ref{alg:multi-path}]
\label{cor:multi-path-complexity}
Let $P$ denote the set of paths between the selected source-destination pairs and define the total
path length as
\[
    M \;=\; \sum_{\pi \in P} |\pi|,
\]
where $|\pi|$ is the hop count of path $\pi$.
Algorithm~\ref{alg:multi-path} runs in $O(M)$ time.
\end{corollary}

\begin{proof}
Pass~1 iterates over every path $\pi \in P$ and, for each link
$e \in \pi$, performs a single comparison and conditional update of
$\mathrm{min\_val}[e]$, each costing $O(1)$.
The total number of (path, link) pairs visited is exactly $M$, so Pass~1
costs $O(M)$.
Pass~2 likewise iterates over every (path, link) pair once and performs a
single conditional write per pair, costing $O(M)$.
All other operations (initialisation, return) are at most $O(M)$.
Hence, the overall running time is $O(M)$.
\end{proof}
\begin{remark}[Explanation of the minimum selection rule]
\label{remark:multi-path-explanation}
Pass~1 scans every (path, link) incidence once, so that
$\mathrm{min\_val}[e] = \min_{\pi \in P:\, e \in \pi} L_{\{\pi, e\}}$ for every
link. Pass~2 lets each path read this precomputed array and write only its own
allocations, which makes it parallel across paths and resolves all shared-link
conflicts in a single round with no iterative coordination. Shared links are
thereby never driven beyond any traversing path's feasible region, and
exclusive links retain their QLO values unchanged. The rule is not globally
optimal, but it is deterministic and inexpensive, and serves as the benchmark
against which GRO is evaluated.
\end{remark}


\begin{thebibliography}{1}
\bibliographystyle{IEEEtran}

\bibitem{vk2025quantuminternet} Kumar, Vinay, et al. "Quantum internet: Technologies, protocols, and research challenges." International Journal of Networked and Distributed Computing 13.2 (2025): 22.

\bibitem{rohde2025quantuminternet} Rohde, Peter P., et al. "The Quantum Internet (Technical Version)." arXiv preprint arXiv:2501.12107 (2025).

\bibitem{vkthesis} Kumar, Vinay. "Making Quantum Networks Work: Routing, Calibration, and Programmable Quantum Repeaters." arXiv preprint arXiv:2606.22316 (2026).

\bibitem{meddeb2025quantuminternet} Meddeb, Aref. "Quantum internet building blocks state of research and development." Computer Networks 261 (2025): 111151.

\bibitem{lalindis25} Lal, Nijil, et al. "Indistinguishability of arbitrary photons for entanglement generation in scalable quantum networks." Optics Express 33.26 (2025): 54501-54512.

\bibitem{zukowski1993swapping} Zukowski, Marek, et al. "" Event-ready-detectors" Bell experiment via entanglement swapping." Physical review letters 71.26 (1993).

\bibitem{pompili2021multinode} Pompili, Matteo, et al. "Realization of a multinode quantum network of remote solid-state qubits." Science 372.6539 (2021): 259-264.

\bibitem{viktor2023trappedionentanglement} Krutyanskiy, Viktor, et al. "Entanglement of trapped-ion qubits separated by 230 meters." Physical Review Letters 130.5 (2023): 050803.

\bibitem{knaut2024nanophotonicentanglement} Knaut, Can M., et al. "Entanglement of nanophotonic quantum memory nodes in a telecom network." Nature 629.8012 (2024): 573-578.

\bibitem{craddock2024newyork} Craddock, Alexander N., et al. "Automated distribution of polarization-entangled photons using deployed New York City fibers." PRX Quantum 5.3 (2024): 030330.

\bibitem{hermans2022teleportation} Hermans, S. L. N., et al. "Qubit teleportation between non-neighbouring nodes in a quantum network." Nature 605.7911 (2022): 663-668.

\bibitem{vk2025endtoend} Kumar, Vinay, et al. "Routing in quantum networks with end‐to‐end knowledge." IET Quantum Communication 6.1 (2025): e70000.

\bibitem{vk2024mixed} Kumar, Vinay, et al. "Routing in Quantum Repeater Networks with Mixed Efficiency Figures." 2024 IEEE Future Networks World Forum (FNWF). IEEE, 2024.

\bibitem{pant2019routing} Pant, Mihir, et al. "Routing entanglement in the quantum internet." npj Quantum Information 5.1 (2019): 25.

\bibitem{inestabuffer25} Iñesta, Álvaro G., et al. "Entanglement buffering with multiple quantum memories." arXiv preprint arXiv:2502.20240 (2025).

\bibitem{durpurification99} Dür, Wolfgang, et al. "Quantum repeaters based on entanglement purification." Physical Review A 59.1 (1999): 169.

\bibitem{peranic2023polcomp} Peranić, Matej, et al. "A study of polarization compensation for quantum networks." EPJ quantum technology 10.1 (2023): 30.

\bibitem{dowling2023nonlocalpol} Dowling, Evan, et al. "Non-local polarization alignment and control in fibers using feedback from correlated measurements of entangled photons." Optics Express 31.2 (2023): 2316-2329.

\bibitem{zhou2025polcomp} Zhou, Xiao-yan, et al. "Efficient polarization-entangled state compensation in quantum entanglement distribution." Optics Express 33.11 (2025): 23204-23213.

\bibitem{chung2022ieqnet} Chung, Joaquin, et al. "Design and implementation of the Illinois Express quantum metropolitan area network." IEEE Transactions on Quantum Engineering 3 (2022): 1-20.

\bibitem{schon2024quantnet} Schon, Damian, et al. "The QUANT-NET testbed development and preliminary results." 2024 IEEE International Conference on Quantum Computing and Engineering (QCE). Vol. 1. IEEE, 2024.

\bibitem{islam2025arqnet} Islam, Md Shariful, et al. "Experimental Demonstration of Software-Orchestrated Quantum Network Applications over a Campus-Scale Testbed." arXiv preprint arXiv:2511.01247 (2025).

\bibitem{azuma2023rmp} Azuma, Koji, et al. "Quantum repeaters: From quantum networks to the quantum internet." Reviews of Modern Physics 95.4 (2023): 045006.

\bibitem{xiao2024psc} Xiao, Zirui, et al. "Purification scheduling control for throughput maximization in quantum networks." Communications Physics 7.1 (2024): 307.

\bibitem{zhang2025linkconfig} Zhang, Qiaolun, et al. "Link configuration for fidelity-constrained entanglement routing in quantum networks." IEEE INFOCOM 2025-IEEE Conference on Computer Communications. IEEE, 2025.

\bibitem{vardoyan2022qnum} Vardoyan, Gayane, and Stephanie Wehner. "Quantum network utility maximization." 2023 IEEE International Conference on Quantum Computing and Engineering (QCE). Vol. 1. IEEE, 2023.

\bibitem{panigrahy2025distributed} Panigrahy, Nitish K., et al. "A Framework for Distributed Resource Allocation in Quantum Networks." arXiv preprint arXiv:2510.09371 (2025).

\bibitem{wang2025fair} Wang, Zhaozhen, et al. "Quantum network optimization: From optimal routing to fair resource allocation." Proceedings of the ACM on Measurement and Analysis of Computing Systems 9.2 (2025): 1-26.

\bibitem{gauthier2024resourceallocation} Gauthier, Scarlett, Thirupathaiah Vasantam, and Gayane Vardoyan. "An on-demand resource allocation algorithm for a quantum network hub and its performance analysis." 2024 IEEE International Conference on Quantum Computing and Engineering (QCE). Vol. 1. IEEE, 2024.

\bibitem{ercetin2026fidelityage} Ercetin, Ozgur, and Zafer Gedik. "Fidelity-Age-Aware Scheduling in Quantum Repeater Networks." arXiv preprint arXiv:2602.09562 (2026).

\bibitem{akaike1974aic} Akaike, Hirotugu. "A new look at the statistical model identification." IEEE transactions on automatic control 19.6 (1974): 716-723.

\bibitem{huistokes2009} Rongqing Hui, Maurice O'Sullivan,
Chapter 2 - Basic Instrumentation for Optical Measurement,
Fiber Optic Measurement Techniques,
Academic Press,
2009,
Pages 129-258,
ISBN 9780123738653,
https://doi.org/10.1016/B978-0-12-373865-3.00002-1.


\bibitem{burnham2002aic} Burnham, Kenneth P., and David R. Anderson, eds. Model selection and multimodel inference: a practical information-theoretic approach. New York, NY: Springer New York, 2002.
\end{thebibliography}
\end{document}